\documentclass[12pt]{article}
\usepackage{amsmath}
\usepackage{txfonts,bm}  
\usepackage{cite}  
\usepackage{graphicx}
\usepackage{color}
\newtheorem{theorem}{Theorem}[section]
\newtheorem{proposition}[theorem]{Proposition}
\newtheorem{lemma}[theorem]{Lemma}
\newtheorem{remark}[theorem]{Remark}
\newcommand{\qed}{\qquad$\square$}

\makeatletter
\@addtoreset{equation}{section}
\renewcommand{\theequation}{\@arabic\c@section.\@arabic\c@equation}
\long\def\@makecaption#1#2{
 \vskip 10pt 
 \setbox\@tempboxa\hbox{#1. #2}
 \ifdim \wd\@tempboxa >\hsize #1. #2\par \else \hbox
to\hsize{\hfil\box\@tempboxa\hfil} 
 \fi}
\makeatother
\begin{document}
\begin{center}
\renewcommand{\baselinestretch}{1.3}\selectfont
\begin{Large}
\textbf{\boldmath Projective reduction of the discrete Painlev\'e system of type $(A_2+A_1)^{(1)}$}
\end{Large}\\[4mm]
\renewcommand{\baselinestretch}{1}\selectfont
\textrm{\large Kenji Kajiwara,
Nobutaka Nakazono,
and Teruhisa Tsuda}\\[2mm]
28 September. 2009 
~Revised: 14 April. 2010
\end{center}
\begin{abstract}
We consider the $q$-Painlev\'e III equation arising from
the birational representation of the affine Weyl group
of type $(A_2+A_1)^{(1)}$.
We study the reduction of the $q$-Painlev\'e III equation
to the $q$-Painlev\'e II equation from the viewpoint of
affine Weyl group symmetry.
In particular, the mechanism of apparent inconsistency
between the hypergeometric solutions to both equations
is clarified by using factorization of difference operators
and the $\tau$ functions.
\end{abstract}
\section{Introduction}
The discrete Painlev\'e equations have been studied actively 
from various points of view. Together with the Painlev\'e 
equations, they are now regarded as one of the most important 
classes of equations in the theory of integrable systems 
(see, for example, \cite{GR:review}).  Originally, the discrete
Painlev\'e equations had been identified as single second-order
equations\cite{Brezin-Kazakov,Periwal-Shevitz,Fokas-Its-Kitaev,
Douglas-Shenker,RGH:dP} and then were generalized to simultaneous 
first-order equations. A typical example is the following equation 
known as a discrete Painlev\'e II equation\cite{Periwal-Shevitz,RGH:dP}:
\begin{equation}\label{sdP2:eqn}
 x_{n+1}+x_{n-1}=\frac{(an+b)x_n+c}{1-{x_n}^2},
\end{equation}
where $x_n$ is the dependent variable, $n$ is the independent variable, 
and $a$, $b$, $c$ $\in\mathbb{C}$ are parameters.
By applying the singularity
confinement criterion\cite{GRP:SC}, (\ref{sdP2:eqn}) is generalized to
\begin{equation}\label{adP2:eqn1}
 x_{n+1}+x_{n-1}=\frac{(an+b)x_n+c+(-1)^nd}{1-{x_n}^2},
\end{equation}
where $d$ is a parameter, with its integrability preserved. Introducing the 
dependent variables $X_n$ and $Y_n$ by 
\begin{equation}\label{dP2:specialization}
 X_n=x_{2n},\quad
 Y_n=x_{2n-1},
\end{equation}
then (\ref{adP2:eqn1}) can be rewritten as
\begin{equation}\label{adP2:eqn}
 Y_{n+1}+Y_{n}=\frac{(2an+b)X_n+c+d}{1-{X_n}^2},\quad 
 X_{n+1}+X_{n}=\frac{(a(2n+1)+b)Y_{n+1}+c-d}{1-{Y_{n+1}}^2}.
\end{equation}
Equation (\ref{adP2:eqn}) is known as a discrete Painlev\'e III equation since 
it admits a continuous limit to the Painlev\'e III equation\cite{GNPRS:dP3}. 
Conversely, (\ref{sdP2:eqn}) can be recovered from (\ref{adP2:eqn}) by putting 
$d=0$ and (\ref{dP2:specialization}). We call this procedure ``{\it symmetrization}'' 
of (\ref{adP2:eqn}), which comes from the terminology of the Quispel--Roberts--Thompson 
(QRT) mapping\cite{QRT1,QRT2}. After this terminology, (\ref{adP2:eqn}) is sometimes 
called the ``{\it asymmetric}'' discrete Painlev\'e II equation, and (\ref{sdP2:eqn}) 
is called the ``{\it symmetric}'' discrete Painlev\'e III equation\cite{KTGR:asymmetric}.

It looks that the symmetrization is a simple specialization of parameters at the 
level of the equation, but some strange phenomena have been reported as to their 
particular solutions expressed in terms of hypergeometric functions 
({\it hypergeometric solutions}). The hypergeometric solutions 
to (\ref{sdP2:eqn}) have been constructed as follows \cite{K:dP2,KOSGR:dP2}:
\begin{proposition}\label{prop:sdP2}
For each $N\in\mathbb{N}$, let $\tau_N^{n}$ be an $N\times N$ determinant defined by
\begin{equation}\label{sdP2:tau}
 \tau_N^n=
  \begin{vmatrix}
   H_{n}&H_{n+2}&\cdots&H_{n+2N-2}\\
   H_{n+1}&H_{n+3}&\cdots&H_{n+2N-1}\\
   \vdots&\vdots&\ddots&\vdots\\
   H_{n+N-1}&H_{n+N+1}&\cdots&H_{n+3N-3}
  \end{vmatrix},
\end{equation}
where $H_n$ is a function satisfying the three-term relation$:$
\begin{equation}\label{dP2:hyper_entries}
 H_{n+1}-zH_n+nH_{n-1}=0.
\end{equation}
Then, 
\begin{equation}
 x_n=\frac{2}{z}~\frac{\tau_{N+1}^{n+1}\tau_N^n}{\tau_{N+1}^n\tau_N^{n+1}}-1,
\end{equation} 
satisfies {\rm (\ref{sdP2:eqn})} with the parameters
\begin{equation}
 a=\frac{8}{z^2},\quad
 b=\frac{4(1+2N)}{z^2},\quad
 c=-\frac{4(1+2N)}{z^2}.
\end{equation}
\end{proposition}
On the other hand, since (\ref{adP2:eqn}) appears as the B\"acklund transformation 
of the Painlev\'e V equation\cite{Ohta:RIMS_dP,ROSG:dP2_qP3}, its hypergeometric 
solutions are essentially the same as those to the Painlev\'e V 
equation\cite{Masuda:p5,Okamoto:p5}. The explicit form of the
hypergeometric solutions to (\ref{adP2:eqn}) are given as follows:
\begin{proposition}\label{prop:adP2}
For each $N\in\mathbb{N}$, let $\tau_N^{n,m}$ be an $N\times N$ determinant
defined by
\begin{equation}\label{adP2:tau}
 \tau_N^{n,m}=
  \begin{vmatrix}
   K_{n}^{m}&K_{n+1}^{m}&\cdots&K_{n+N-1}^{m}\\
   K_{n+1}^{m}&K_{n+2}^{m}&\cdots&K_{n+N}^{m}\\
   \vdots&\vdots&\ddots&\vdots\\
   K_{n+N-1}^{m}&K_{n+N}^{m}&\cdots&K_{n+2N-2}^{m}
  \end{vmatrix},
\end{equation}
where $K_n^{m}$ is a function satisfying 
\begin{equation}\label{adP2:hyper_entries}
 K_{n+1}^m-K_{n}^{m}-tK_{n+1}^{m+1}=0,\quad
 nK_{n+1}^m-(n+t)K_{n}^m-(n-m)tK_{n}^{m+1}=0.
\end{equation}
Then,
\begin{equation}
 X_n=2(n+2N-1)\frac{\tau_{N+1}^{n,m}\tau_{N}^{n,m}}
       {\tau_{N+1}^{n-1,m-1}\tau_{N}^{n+1,m+1}}-1,\quad
 Y_n=\frac{2}{t}\frac{\tau_{N+1}^{n-1,m-1}\tau_N^{n,m+1}}
       {\tau_{N+1}^{n-1,m}\tau_N^{n,m}}-1,
\end{equation}
satisfy {\rm (\ref{adP2:eqn})} with the parameters
\begin{equation}
 a=-\frac{4}{t},\quad
 b=\frac{-4(-m+2N-1)}{t},\quad
 c=\frac{2(1+2N)}{t},\quad
 d=\frac{2(2m+2N-3)}{t}.
\end{equation}
\end{proposition}
It is obvious that substituting $d=0$ into the hypergeometric solutions to
(\ref{adP2:eqn}) in Proposition \ref{prop:adP2} do not yield those to 
(\ref{sdP2:eqn}) in Proposition \ref{prop:sdP2}. In particular, we remark 
the following differences between the two solutions:
\begin{description}
 \item[{\rm (i)}] 
 the hypergeometric functions are different. Equation
 (\ref{dP2:hyper_entries}) can be solved by the parabolic 
 cylinder function (Weber function), while (\ref{adP2:hyper_entries}) 
 can be solved by the confluent hypergeometric function. In fact,
 the former function is expressed as a specialization of the latter, 
 but this specialization is not consistent with the symmetrization;
 \item[{\rm (ii)}] 
 structures of the determinant are different. The determinant 
 (\ref{sdP2:tau}) has asymmetry in the shift of index: the shift 
 in the vertical direction is one while that in
 the horizontal direction is two. 
 On the other hand, the determinant
 (\ref{adP2:tau}) is an ordinary Hankel determinant.
\end{description}
We note that similar phenomena have been reported also for some other discrete 
Painlev\'e equations\cite{KOS:dp3,NKT:qp2,HKW:A1A1}. Many integrable systems 
admit particular solutions expressed in terms of determinants, but such an 
asymmetric structure of the determinant solutions has been seen only in the 
hypergeometric solutions to the discrete Painlev\'e equations.
Note here that these phenomena cannot be seen for the algebraic (or rational) 
solutions. For example, it is known that substituting $d=0$ into the determinant 
expression of the rational solutions to (\ref{adP2:eqn}) yields those to 
(\ref{sdP2:eqn}); see \cite{KYO:dP2_rational,MOK:p5_RIMS,MOK:p5_Nagoya}.

The $\tau$ function is one of the most important objects in the theory of 
integrable systems and is regarded as carrying the underlying fundamental 
mathematical structures. Concerning the discrete Painlev\'e equations, 
investigation of the $\tau$ functions started \cite{KOS:dp3,KOSGR:dP2} through
the search for the explicit formulae of the hypergeometric and algebraic solutions. 
In fact, the above mysterious asymmetric structure has been one motivation of 
further study.

It is now known that theory of birational representations of affine Weyl groups 
provides us with an algebraic tool to study the Painlev\'e
systems\cite{Noumi:book,Okamoto:p24,Okamoto:p6,Okamoto:p5,Okamoto:p3}. 
Moreover, a geometric framework of the two-dimensional Painlev\'e systems has been 
presented based on certain rational surfaces\cite{KMNOY:Cremona,Sakai:Painleve}. 
Combining these results enables us to study the Painlev\'e systems effectively. 
For instance, it played a crucial role in the identification of
hypergeometric functions that appear as the particular solutions to the Painlev\'e 
systems in Sakai's classification\cite{KMNOY:elliptic,KMNOY:hyper1,KMNOY:hyper2}.

The purpose of this paper is to clarify the mechanism of the phenomena of hypergeometric 
solutions from the viewpoint of the affine Weyl group symmetry.  We shall take the 
$q$-Painlev\'e equation of type $(A_2+A_1)^{(1)}$ as an example, which is the simplest 
non-trivial discrete Painlev\'e system\cite{Sakai:Painleve}. 
The key is to formulate the symmetrization 
in terms of the birational representation of the affine
Weyl group, where the discrete Painlev\'e equation arises from the action of the translational
subgroup. In fact, the discrete time evolution of the symmetric case comes from
a ``half-step'' of a translation of 
the affine Weyl group through a restriction to a certain line in the
parameter space.
Conversely, we can derive various discrete 
Painlev\'e equations from elements of infinite order
that are not necessarily translations by taking a 
projection on a certain subspace of the parameters.
We call such a procedure to obtain a ``smaller'' 
discrete time evolution of Painlev\'e type a
{\it projective reduction}.

This paper is organized as follows: 
in Section 2, we introduce a $q$-Painlev\'e III equation and
derive a $q$-Painlev\'e II equation by applying the symmetrization. 
Then we give a brief review on
their hypergeometric solutions. In Section 3, 
we first introduce the family of B\"acklund
transformations of the $q$-Painlev\'e III equation, 
which is a birational representation of the
affine Weyl group of type $(A_2+A_1)^{(1)}$. 
We next lift the representation on the level of
$\tau$ functions and derive various bilinear equations. 
We then clarify the mechanism of the
inconsistency among the hypergeometric solutions by using this framework. 
Some concluding remarks
are given in Section 4.

{\it Note}.
We use the following conventions of $q$-analysis
throughout this paper.\\
$q$-Shifted factorials:
\begin{equation}
 (a;q)_k=\prod_{i=1}^{k}(1-aq^{i-1}).
\end{equation}
Basic hypergeometric series\cite{Gasper-Rahman:BHS}:
\begin{equation}
 {}_1\varphi_1\left(\begin{matrix}a\\b\end{matrix};q,z\right)
  =\sum_{k=0}^\infty \frac{(a;q)_k}{(b;q)_k(q;q)_k}
  (-1)^kq^{k(k-1)/2}z^k.
\end{equation}
Jacobi theta function:
\begin{equation}
 \Theta(a;q)=(a;q)_\infty(qa^{-1};q)_\infty.
\end{equation}
Elliptic gamma function:
\begin{align}
 \Gamma(a;p,q)=\cfrac{(q^2a^{-1};p,q)_{\infty}}{(a;p,q)_{\infty}},
\end{align}
where
\begin{equation}
 (a;p,q)_k=\prod_{i,j=0}^{k-1} (1-p^iq^ja).
\end{equation}
It holds that
\begin{align}
 &\Theta(qa;q)=-a^{-1}\Theta(a;q),\\
 &\Gamma(qa;q,q)=\Theta(a;q)\Gamma(a;q,q).
\end{align}
\section{${\bm q}$-P$_{\rm\bf III}$ and ${\bm q}$-P$_{\rm\bf II}$}
We consider the following system of 
$q$-difference equations\cite{KK:qp3,KNY:qp4,Sakai:Painleve}:
\begin{equation}\label{qp3:eqn}
 g_{n+1}=\frac{q^{2N+1}c^2}{f_ng_n}~
  \frac{1+a_0q^{n}f_n}{a_0q^n+f_n},\quad
 f_{n+1}=\frac{q^{2N+1}c^2}{f_ng_{n+1}}~
  \frac{1+a_2a_0q^{n-m}g_{n+1}}{a_2a_0q^{n-m}+g_{n+1}},
\end{equation}
for the unknown functions $f_n=f_n(m,N)$ and $g_n=g_n(m,N)$ and 
the independent variable $n\in\mathbb{Z}$.
Here $m,N\in\mathbb{Z}$ and $a_0,a_2,c,q\in\mathbb{C}^{\times}$ are parameters.
Equation (\ref{qp3:eqn}) has the (extended) affine Weyl group symmetry of type
$(A_2+A_1)^{(1)}$ and is known as a $q$-Painlev\'e III equation ($q$-P$_{\rm
III}$) since the continuous limit yields the Painlev\'e III equation.
We also consider the following $q$-difference equation\cite{RG:coales,NKT:qp2}:
\begin{equation}\label{qp2:eqn}
 X_{k+1}=\frac{q^{2N+1}c^2}{X_{k}X_{k-1}}
 ~\frac{1+a_0q^{k/2}X_k}{a_0q^{k/2}+X_k},
\end{equation}
for the unknown function $X_k=X_k(N)$ and the independent variable $k\in\mathbb{Z}$.
Equation (\ref{qp2:eqn}) is a $q$-Painlev\'e II equation ($q$-P$_{\rm II}$) 
and actually it admits
a continuous limit to the Painlev\'e II equation.

Note that substituting
\begin{equation}\label{qp3:symmetrization1}
 m=0,\quad
 a_2=q^{1/2},
\end{equation}
and putting
\begin{equation}\label{qp3:symmetrization2}
 f_k(0,N)=X_{2k}(N),\quad
 g_k(0,N)=X_{2k-1}(N),
\end{equation}
in {\rm (\ref{qp3:eqn})} yield {\rm (\ref{qp2:eqn})}.

We shall briefly review the hypergeometric solutions to
$q$-P$_{\rm III}$ and $q$-P$_{\rm II}$ following \cite{KK:qp3,NKT:qp2}
and then compare their structures.
\subsection{Hypergeometric solutions to ${\bm q}$-P$_{\rm\bf III}$}
First, we review the hypergeometric solutions to $q$-P$_{\rm III}$.
For each $N\in\mathbb{Z}_{\geq 0}$, let $\psi^{n,m}_N$ be 
an $N\times N$ determinant defined by
\begin{equation}\label{qp3:det}
 \psi^{n,m}_N=
  \begin{vmatrix}
   F_{n,m}& F_{n+1,m}&\cdots&F_{n+N-1,m}\\
   F_{n-1,m}& F_{n,m}&\cdots&F_{n+N-2,m}\\
   \vdots&\vdots &\ddots&\vdots\\
   F_{n-N+1,m}&F_{n-N+2,m}&\cdots&F_{n,m}
  \end{vmatrix},\quad
 \psi^{n,m}_0=1,
\end{equation}
where $F_{n,m}$ satisfies 
\begin{equation}\label{qp3:contiguity}
 \begin{array}{l}\medskip
  F_{n+1,m} - F_{n,m} = -{a_0}^2q^{2n} F_{n,m-1},\\
  F_{n,m+1} - F_{n,m} = -{a_2}^{-2}q^{2m+2} F_{n-1,m}.
 \end{array}
\end{equation}
\begin{lemma}[\cite{KK:qp3}]\label{lem:qp3_bl}
$\psi^{n,m}_N$ satisfies the following bilinear difference equations$:$
\begin{align}
 &{a_0}^2q^{2n-2}\psi^{n-1,m-1}_{N+1}\psi^{n,m}_{N}
  -q^{2N}\psi^{n,m-1}_{N}\psi^{n-1,m}_{N+1}
  +\psi^{n-1,m-1}_{N}\psi^{n,m}_{N+1}=0,\label{qp3:bl1}\\
 &\psi^{n,m}_{N+1}\psi^{n,m-1}_{N}
  -q^{-2N}\psi^{n-1,m-1}_{N}\psi^{n+1,m}_{N+1}
  -{a_0}^2q^{2n}\psi^{n,m}_{N}\psi^{n,m-1}_{N+1}=0,\label{qp3:bl2}\\
 &\psi^{n,m}_{N+1}\psi^{n-1,m-1}_{N}
  -\psi^{n,m-1}_{N+1}\psi^{n-1,m}_{N}
  +{a_2}^{-2}q^{2m}\psi^{n,m}_{N}\psi^{n-1,m-1}_{N+1}=0,\label{qp3:bl3}\\
 &\psi^{n,m-1}_{N+1}\psi^{n,m}_{N}
  -{a_2}^{-2}q^{2m}\psi^{n-1,m-1}_{N+1}\psi^{n+1,m}_{N} 
  -\psi^{n,m-1}_{N}\psi^{n,m}_{N+1}=0.\label{qp3:bl4}
\end{align}
\end{lemma}
\begin{proposition}[\cite{KK:qp3}]\label{prop:qp3_hyper}
The hypergeometric solutions to $q$-{\rm P}$_{\rm III}$, $(\ref{qp3:eqn})$, with $c=1$ are given by
\begin{equation}
 f_n=-a_0q^n\frac{\psi^{n,m-1}_{N+1}\psi^{n,m}_{N}}
  {\psi^{n,m}_{N+1}\psi_{N}^{n,m-1}},\quad
 g_n={a_0}^{-1}a_2q^{-n-m+1}\frac{\psi^{n,m}_{N+1}\psi^{n-1,m-1}_{N}}
  {\psi^{n-1,m-1}_{N+1}\psi^{n,m}_{N}}.
\end{equation}
\end{proposition}
Proposition \ref{prop:qp3_hyper} follows from Lemma \ref{lem:qp3_bl}.
\begin{remark}\label{remark:qp3_hyper}\rm
\begin{enumerate}
\item 
The general solution to (\ref{qp3:contiguity}) is given by
\begin{align}
 F_{n,m}=&\frac{A_{n,m}}{({a_2}^{-2}q^{2m+2};q^2)_\infty}
  {}_1\varphi_1\left(\begin{matrix}0\\ {a_2}^{2}q^{-2m}\end{matrix}
  ;q^{2},{a_2}^2{a_0}^2q^{2n-2m}\right)\nonumber\\
 &+B_{n,m}\cfrac{\Theta({a_0}^2{a_2}^2q^{2n-2m-2};q^2)}
  {({a_2}^2q^{-2m-2};q^2)_\infty\Theta({a_0}^2q^{2n};q^2)}
  {}_1\varphi_1\left(\begin{matrix}0\\ {a_2}^{-2}q^{2m+4}\end{matrix}
  ;q^{2},{a_0}^2q^{2n+2}\right),\label{qp3:entry}
\end{align}
where $A_{n,m}$ and $B_{n,m}$ are periodic functions of period one
for $n$ and $m$, i.e., 
\begin{equation}
 A_{n,m}=A_{n+1,m}=A_{n,m+1},\quad B_{n,m}=B_{n+1,m}=B_{n,m+1}.
\end{equation}
Note that $F_{n,m}$ satisfies the three-term relation with respect to $n$:
\begin{equation}\label{qp3:3-term}
 F_{n+1,m}+\left({a_0}^2q^{2n}-{a_2}^{-2}q^{2m+2}-1\right)F_{n,m}
 +{a_2}^{-2}q^{2m+2}F_{n-1,m}=0.
\end{equation}
\item $\psi^{n,m}_{N}$ satisfies the discrete Toda equation:
\begin{equation}\label{dToda:bl} 
 \psi^{n,m}_{N+1}\psi^{n,m}_{N-1}-\left(\psi^{n,m}_{N}\right)^2
 +\psi^{n+1,m}_{N}\psi^{n-1,m}_{N}=0.
\end{equation}
In general, (\ref{dToda:bl}) admits a solution 
expressed in terms of the Toeplitz type determinant
\begin{equation}
 \psi^{n,m}_{N}=\det\left(c_{n-i+j,m}\right)_{i,j=1,\ldots,N}\quad (N>0),
\end{equation}
for an arbitrary function $c_{n,m}$ under the boundary conditions
\begin{equation}\label{boundary:eqn}
 \psi^{n,m}_0=1,\quad
 \psi^{n,m}_N=0\quad (N<0).
\end{equation}
Since the hypergeometric solutions to 
$q$-P$_{\rm III}$ satisfy the conditions (\ref{boundary:eqn}),
the bilinear equation (\ref{dToda:bl}) is 
regarded as to fix the determinant structure of the solutions.
\end{enumerate}
\end{remark}
\subsection{Hypergeometric solutions to ${\bm q}$-P$_{\rm\bf II}$}
Next, we review the hypergeometric solutions to $q$-P$_{\rm II}$.
For each $N\in\mathbb{Z}_{\geq 0}$, let $\phi^k_N$ be an $N\times N$ determinant defined by
\begin{equation}\label{qp2:det}
\phi^k_N=
 \begin{vmatrix}
  G_{k}&G_{k+2}&\cdots&G_{k+2N-2}\\
  G_{k-1}&G_{k+1}&\cdots&G_{k+2N-3}\\
  \vdots&\vdots&\ddots&\vdots\\
  G_{k-N+1}&G_{k-N+3}&\cdots&G_{k+N-1}
 \end{vmatrix}
,\quad
 \phi^k_0=1, 
\end{equation}
where $G_k$ satisfies
\begin{equation}\label{qp2:3-term}
 G_{k+1} - G_k + \frac{1}{{a_0}^2q^{k}}G_{k-1}=0.
\end{equation}
\begin{lemma}[\cite{NKT:qp2}]\label{lem:qp2_bl}
 $\phi^k_N$ satisfies the following bilinear difference equations$:$
\begin{align}
 &{a_0}^{-2}q^{-k+1}\phi^{k-2}_{N+1}\phi^{k+1}_{N}+\phi^{k}_{N+1}\phi^{k-1}_{N}
 -q^{-N}\phi^{k-1}_{N+1}\phi^{k}_{N}=0,\label{qp2:bl1}\\
 &q^{N}\phi^{k+1}_{N+1}\phi^{k-2}_{N}+{a_0}^{-2}q^{-k-N}\phi^{k-1}_{N+1}\phi^{k}_{N}
 -\phi^{k}_{N+1}\phi^{k-1}_{N}=0.\label{qp2:bl2}
\end{align}
\end{lemma}
\begin{proposition}[\cite{NKT:qp2}]\label{prop:qp2_hyper}
The hypergeometric solutions to $q$-{\rm P}$_{\rm II}$, 
$(\ref{qp2:eqn})$, with $c=1$ are given by
\begin{equation}\label{eqn:X_phi}
 X_k=-a_0q^{(k+2N)/2}
 \frac{\phi^k_{N+1}\phi^{k-1}_{N}}{\phi^{k-1}_{N+1}\phi^k_{N}}.
\end{equation}
\end{proposition}
Proposition \ref{prop:qp2_hyper} follows from Lemma \ref{lem:qp2_bl}.
\begin{remark}\label{remark:qp2_hyper}\rm
\begin{enumerate}
\item The general solution to (\ref{qp2:3-term}) is given by
\begin{equation}\label{qp2:Gk}
\begin{split}
 G_k=&
 A_k\Theta(ia_0q^{(2k+1)/4};q^{1/2})
  {}_1\varphi_1\left(\begin{matrix}0\\-q^{1/2}\end{matrix}
  ;q^{1/2},-ia_0q^{(3+2k)/4}\right)\\
 &+B_k\Theta(-ia_0q^{(2k+1)/4};q^{1/2})
  {}_1\varphi_1\left(\begin{matrix}0\\-q^{1/2}\end{matrix}
  ;q^{1/2},ia_0q^{(3+2k)/4}\right),
\end{split}
\end{equation}
where $A_k$ and $B_k$ are periodic functions of period one, i.e., 
\begin{equation}
 A_k=A_{k+1},\quad B_k=B_{k+1}.
\end{equation}
 \item $\phi^k_N$ also satisfies the bilinear equation
\begin{equation}\label{dToda2:bl}
 \phi^{k}_{N+1}\phi^{k+1}_{N-1}-\phi^{k}_{N}\phi^{k+1}_{N}
 +\phi^{k+2}_{N}\phi^{k-1}_{N} = 0,
\end{equation}
which is a variant of the discrete Toda equation. Under the conditions
\begin{equation}
 \phi^{k}_0=1,\quad
 \phi^{k}_N=0\quad (N<0),
\end{equation}
(\ref{dToda2:bl}) admits a solution expressed by
\begin{equation}
 \phi^{k}_N=\det\left(c_{k+2i-j-1}\right)_{i,j=1,\ldots,N}\quad (N>0),
\end{equation}
for an arbitrary function $c_k$. Hence, (\ref{dToda2:bl}) can be regarded as the
bilinear equation that fixes the determinant structure of the hypergeometric solutions to
$q$-P$_{\rm II}$.
\end{enumerate}
\end{remark}
\subsection{Comparison of the hypergeometric solutions}
By comparing the hypergeometric solution to $q$-P$_{\rm III}$ 
with that to $q$-P$_{\rm II}$
(see Propositions \ref{prop:qp3_hyper} and \ref{prop:qp2_hyper}) one may
immediately notice that a na\"ive application of the 
specialization (\ref{qp3:symmetrization1}) 
to the former does not yield the latter.
As analogous to the phenomena seen in Section 1, we
find the following differences between the two solutions:
\begin{description}
 \item[{\rm (i)}] the hypergeometric functions are different. In fact, substituting 
$a_2=q^{1/2}$ into (\ref{qp3:3-term}) and (\ref{qp3:entry}) do not yield 
(\ref{qp2:3-term}) and (\ref{qp2:Gk}), respectively;
 \item[{\rm (ii)}] the determinant structures are different. 
\end{description}

Besides the determinant formula for the hypergeometric solution to 
$q$-P$_{\rm II}$ in Proposition \ref{prop:qp2_hyper},
one can also obtain another formula from that to $q$-P$_{\rm III}$
in Proposition \ref{prop:qp3_hyper} through a 
specialization (\ref{qp3:symmetrization1}).
We set
\begin{equation}
 f_n=X_{2n},\quad
 g_n=X_{2n-1},
\end{equation}
and define $\hat{G}_k$ as 
\begin{align}
 \hat{G}_k
 &=\cfrac{1+(-1)^k}{2}\Theta({a_0}^2q^k;q^2)F_{\frac{k}{2},-1}
 +\cfrac{1-(-1)^k}{2}\Theta({a_0}^2q^{k+1};q^2)F_{\frac{k+1}{2},0}\nonumber\\
 &=\begin{cases}
  \Theta({a_0}^2q^{2n};q^2)F_{n,-1}&(k=2n)\\
  \Theta({a_0}^2q^{2n};q^2)F_{n,0}&(k=2n-1)
 \end{cases}
\end{align}
with $F_{n,m}$ given in Remark \ref{remark:qp3_hyper}.
The system (\ref{qp3:contiguity}) reduces to the equation
\begin{equation}\label{qp2:3-term2}
 \hat{G}_{k+1}-\hat{G}_k+\frac{1}{{a_0}^2q^{k}}\hat{G}_{k-1}=0,
\end{equation}
which coincides with (\ref{qp2:3-term}).
Then we have solutions to $q$-P$_{\rm II}$:
\begin{equation}\label{eqn:X_psi}
 X_{2n}=-a_0q^n\frac{\psi^{n,-1}_{N+1}\psi^{n,0}_{N}}
  {\psi^{n,0}_{N+1}\psi_{N}^{n,-1}},\quad
 X_{2n-1}={a_0}^{-1}q^{(-2n+3)/2}
  \frac{\psi^{n,0}_{N+1}\psi^{n-1,-1}_{N}}{\psi^{n-1,-1}_{N+1}\psi^{n,0}_{N}},
\end{equation}
where
\begin{align}
 &\psi^{n,-1}_N=\left(\prod^N_{i=1}\cfrac{1}{\Theta({a_0}^2q^{2n+2i-2};q^2)}\right)
 \cfrac{q^{N(N^2-1)/3}}{(-{a_0}^2q^{2n})^{N(N-1)/2}}
 \begin{vmatrix}
  \hat{G}_{2n}&\hat{G}_{2n+2}&\cdots&\hat{G}_{2n+2N-2}\\
  \hat{G}_{2n-2}&q^{-2}\hat{G}_{2n}&\cdots&q^{-2(N-1)}\hat{G}_{2n+2N-4}\\
  \vdots&\vdots &\ddots&\vdots\\
  \hat{G}_{2n-2N+2}&q^{-2(N-1)}\hat{G}_{2n-2N+4}&\cdots&q^{-2(N-1)^2}\hat{G}_{2n}
 \end{vmatrix},\label{eqn:qp3_qp2_det_1}\\
 &\psi^{n,0}_N=\left(\prod^N_{i=1}\cfrac{1}{\Theta({a_0}^2q^{2n+2i-2};q^2)}\right)
 \cfrac{q^{N(N^2-1)/3}}{(-{a_0}^2q^{2n})^{N(N-1)/2}}
 \begin{vmatrix}
  \hat{G}_{2n-1}&\hat{G}_{2n+1}&\cdots&\hat{G}_{2n+2N-3}\\
  \hat{G}_{2n-3}&q^{-2}\hat{G}_{2n-1}&\cdots&q^{-2(N-1)}\hat{G}_{2n+2N-5}\\
  \vdots&\vdots &\ddots&\vdots\\
  \hat{G}_{2n-2N+1}&q^{-2(N-1)}\hat{G}_{2n-2N+3}&\cdots&q^{-2(N-1)^2}\hat{G}_{2n-1}
 \end{vmatrix}.\label{eqn:qp3_qp2_det_2}
\end{align}
We can conform the shift of indices of 
(\ref{eqn:qp3_qp2_det_1}) and (\ref{eqn:qp3_qp2_det_2})
to that of (\ref{qp2:det}).
Actually, in (\ref{eqn:qp3_qp2_det_1}) and (\ref{eqn:qp3_qp2_det_2}), 
multiplying the $i$-th row by
${a_0}^{-2}q^{-2n+2i-3}$ and ${a_0}^{-2}q^{-2n+2i-2}$, respectively.
Then adding the $(i-1)$-th row to the $i$-th row, 
the indices of $\hat{G}$ in the $i$-th row increase by one because of
(\ref{qp2:3-term2}). 
Repeating this operation, we obtain
\begin{align}
 &\psi^{n,-1}_N
 =(-1)^{-N(N-1)/2}\left(\prod^N_{i=1}\cfrac{1}{\Theta({a_0}^2q^{2n+2i-2};q^2)}\right)
 q^{-N(N-1)(N-5)/6}
 \begin{vmatrix}
  \hat{G}_{2n}&\hat{G}_{2n+2}&\cdots&\hat{G}_{2n+2N-2}\\
  \hat{G}_{2n-1}&\hat{G}_{2n+1}&\cdots&\hat{G}_{2n+2N-3}\\
  \vdots&\vdots&\ddots&\vdots\\
  \hat{G}_{2n-N+1}&\hat{G}_{2n-N+3}&\cdots&\hat{G}_{2n+N-1}
 \end{vmatrix},\\
 &\psi^{n,0}_N
 =(-1)^{-N(N-1)/2}\left(\prod^N_{i=1}\cfrac{1}{\Theta({a_0}^2q^{2n+2i-2};q^2)}\right)
 q^{-N(N-1)(N-2)/6}
 \begin{vmatrix}
  \hat{G}_{2n-1}&\hat{G}_{2n+1}&\cdots&\hat{G}_{2n+2N-3}\\
  \hat{G}_{2n-2}&\hat{G}_{2n}&\cdots&\hat{G}_{2n+2N-4}\\
  \vdots&\vdots&\ddots&\vdots\\
  \hat{G}_{2n-N}&\hat{G}_{2n-N+2}&\cdots&\hat{G}_{2n+N-2}
 \end{vmatrix}.
\end{align}
Thus, we are led to the following unified expression of (\ref{eqn:X_psi}):
\begin{align}
 &X_k=-a_0q^{(k+2N)/2}\frac{\hat{\phi}^{k}_{N+1}\hat{\phi}^{k-1}_{N}}
  {\hat{\phi}^{k-1}_{N+1}\hat{\phi}_{N}^{k}},
 \label{eqn:X_psi_hat}\\
 &\hat{\phi}^{k}_{N}
 =\begin{vmatrix}
  \hat{G}_{k}&\hat{G}_{k+2}&\cdots&\hat{G}_{k+2N-2}\\
  \hat{G}_{k-1}&\hat{G}_{k+1}&\cdots&\hat{G}_{k+2N-3}\\
  \vdots&\vdots&\ddots&\vdots\\
  \hat{G}_{k-N+1}&\hat{G}_{k-N+3}&\cdots&\hat{G}_{k+N-1}
 \end{vmatrix},\\
 &\hat{G}_k
 =\hat{A}_k\cfrac{\Theta({a_0}^2q^k;q^2)}{(q^{-1};q^2)_\infty}
  {}_1\varphi_1\left(\begin{matrix}0\\ q^3\end{matrix}
  ;q^{2},{a_0}^2q^{k+3}\right)
 +\hat{A}_{k+1}\cfrac{\Theta({a_0}^2q^{k+1};q^2)}{(q;q^2)_\infty}
  {}_1\varphi_1\left(\begin{matrix}0\\ q\end{matrix}
  ;q^{2},{a_0}^2q^{k+2}\right).
 \label{eqn:hat_Gk}
\end{align}
Here $\hat{A}_k$ is defined by
\begin{equation}
 \hat{A}_k=A_{\frac{k}{2},0}\cfrac{1+(-1)^k}{2}+B_{\frac{k+1}{2},0}\cfrac{1-(-1)^k}{2}
 =\begin{cases}
  A_{n,0}&(k=2n)\\
  B_{n,0}&(k=2n-1)
 \end{cases}
\end{equation}
and is a periodic function of period two, i.e., 
\begin{equation}
 \hat{A}_k=\hat{A}_{k+2}.
\end{equation}

In fact, both (\ref{qp2:Gk}) and (\ref{eqn:hat_Gk})
give the general solution to the same
equation (\ref{qp2:3-term}) (or equivalently (\ref{qp2:3-term2})).
This fact thus implies the existence of certain
identities among the basic hypergeometric series
${}_1\varphi_1$ with two different bases $q^2$ and $q^{1/2}$;
see Appendix \ref{appendix:A}.

In the rest of this paper,
we shall clarify the difference of hypergeometric solutions to
$q$-P$_{\rm III}$ and $q$-P$_{\rm II}$ 
(see Propositions \ref{prop:qp3_hyper} and \ref{prop:qp2_hyper})
by using the underlying symmetry of an affine Weyl group.
\begin{remark}\rm
The correspondence between the rational solutions 
to $q$-P$_{\rm III}$ (see \cite{K:qp3}) and that to
$q$-P$_{\rm II}$ (see \cite{NKT:qp2}) are straightforward.
It is easily verified that substituting $a_2=q^{1/2}$ into
the former yields the latter.
\end{remark}
\section{Projective reduction from ${\bm q}$-P$_{\rm\bf III}$ to ${\bm q}$-P$_{\rm\bf II}$}
\subsection{Birational representation of $\widetilde{W}((A_2+A_1)^{(1)})$}
We formulate the family of B\"acklund transformations of $q$-P$_{\rm III}$ as a
birational representation of the extended affine Weyl group of type
$(A_2+A_1)^{(1)}$\cite{KK:qp3,KNY:qp4}. 
We refer to \cite{Noumi:book} for basic ideas of this formulation.

We define the transformations $s_i$ ($i=0,1,2$) and $\pi$ on the variables
$f_j$ ($j=0,1,2$) and parameters $a_k$ ($k=0,1,2$) by
\begin{align}
 s_i(a_j) &= a_j{a_i}^{-a_{ij}},
 &&s_i(f_j) = f_j\left(\frac{a_i+f_i}{1+a_if_i}\right)^{u_{ij}},\\
 \pi(a_i) &= a_{i+1},
 &&\pi(f_i) = f_{i+1},
\end{align}
for $i,j\in\mathbb{Z}/3\mathbb{Z}$.
Here the symmetric $3\times 3$ matrix
\begin{equation}
 A=(a_{ij})_{i,j=0}^2
 =\left(\begin{array}{ccc}2&-1&-1\\-1&2&-1\\-1&-1&2\end{array}\right),
\end{equation}
is the Cartan matrix of type $A_2^{(1)}$,
and the skew-symmetric one
\begin{equation}
 U=(u_{ij})_{i,j=0}^2
 =\left(\begin{array}{ccc}0&1&-1\\-1&0&1\\1&-1&0\end{array}\right),
\end{equation}
represents an orientation of the corresponding Dynkin diagram.
We also define the transformations $w_j$ ($j=0,1$) and $r$ by
\begin{align}
 w_0(f_i)
 &=\frac{a_ia_{i+1}(a_{i-1}a_i+a_{i-1}f_i+f_{i-1}f_i)}
  {f_{i-1}(a_ia_{i+1}+a_if_{i+1}+f_if_{i+1})},
 &w_0(a_i)=a_i,\\
 w_1(f_i)
 &=\frac{1+a_if_i+a_ia_{i+1}f_if_{i+1}}
  {a_ia_{i+1}f_{i+1}(1+a_{i-1}f_{i-1}+a_{i-1}a_if_{i-1}f_i)},
 &w_1(a_i) = a_i,\\
 r(f_i)&=\frac{1}{f_i},
 &r(a_i)=a_i,
\end{align}
for $i\in\mathbb{Z}/3\mathbb{Z}$.
\begin{proposition}[\cite{KNY:qp4}]
The group of birational transformations
$\langle s_0,s_1,s_2,\pi, w_0,w_1,r\rangle$
forms the extended affine Weyl group of type $(A_2+A_1)^{(1)}$,
denoted by $\widetilde{W}((A_2+A_1)^{(1)})$. 
Namely, the transformations satisfy the fundamental relations
\begin{equation}
 {s_i}^2=(s_is_{i+1})^3=\pi^3=1,\ 
 \pi s_i = s_{i+1}\pi\
 (i\in\mathbb{Z}/3\mathbb{Z}),\quad
 {w_0}^2={w_1}^2=r^2=1,\ 
 rw_0=w_1r,
\end{equation}
and the actions of $\widetilde{W}(A_2^{(1)})=\langle s_0,s_1,s_2,\pi\rangle$ and 
$\widetilde{W}(A_1^{(1)})=\langle w_0,w_1,r\rangle$ commute with each other.
\end{proposition}
In general, for a function $F=F(a_i,f_j)$, we let an element
$w\in\widetilde{W}((A_2+A_1)^{(1)})$
act as $w.F(a_i,f_j)=F(a_i.w,f_j.w)$, that is, $w$
acts on the arguments from the right.  Note that $a_0a_1a_2=q$ 
and $f_0f_1f_2=qc^2$ are invariant under the actions of 
$\widetilde{W}((A_2+A_1)^{(1)})$ 
and $\widetilde{W}(A_2^{(1)})$,
respectively. We define the translations $T_i$ ($i=1,2,3,4$) by
\begin{equation}
 T_1=\pi s_2s_1,\quad
 T_2=s_1\pi s_2,\quad
 T_3=s_2s_1\pi,\quad
 T_4=rw_0,
\end{equation}
whose actions on parameters $a_i$ $(i=0,1,2)$ and $c$ are given by
\begin{equation}
 \begin{array}{l}
  T_1:~(a_0,a_1,a_2,c)\mapsto(qa_0,q^{-1}a_1,a_2,c),\\
  T_2:~(a_0,a_1,a_2,c)\mapsto(a_0,qa_1,q^{-1}a_2,c),\\
  T_3:~(a_0,a_1,a_2,c)\mapsto(q^{-1}a_0,a_1,qa_2,c),\\
  T_4:~(a_0,a_1,a_2,c)\mapsto(a_0,a_1,a_2,qc).
\end{array}
\end{equation}
Note that $T_i$ ($i=1,2,3,4$) commute with each other and $T_1T_2T_3=1$.
The action of $T_1$ on $f$-variables can be expressed as
\begin{equation}\label{qp3:eqn2}
 T_1(f_1)=\frac{qc^2}{f_1f_0}~\frac{1+a_0f_0}{a_0+f_0},\quad
 T_1(f_0)=\frac{qc^2}{f_0T_1(f_1)}~\frac{1+a_2a_0T_1(f_1)}{a_2a_0+T_1(f_1)}.
\end{equation}
Or, applying ${T_1}^n{T_2}^m{T_4}^N$ $(n,m,N\in\mathbb{Z})$ 
on (\ref{qp3:eqn2}) and putting
\begin{equation}
 f_{i,N}^{n,m}={T_1}^n{T_2}^m{T_4}^N(f_i)\quad (i=0,1,2), 
\end{equation}
we obtain
\begin{equation}\label{qp3:eqn3}
 f_{1,N}^{n+1,m}=\frac{q^{2N+1}c^2}{f_{1,N}^{n,m}f_{0,N}^{n,m}}
  ~\frac{1+a_0q^nf_{0,N}^{n,m}}{a_0q^n+f_{0,N}^{n,m}},\quad
 f_{0,N}^{n+1,m}=\frac{q^{2N+1}c^2}{f_{0,N}^{n,m}f_{1,N}^{n+1}}
  ~\frac{1+a_2a_0q^{n-m}f_{1,N}^{n+1,m}}{a_2a_0q^{n-m}+f_{1,N}^{n+1,m}},
\end{equation}
which is equivalent to $q$-P$_{\rm III}$, (\ref{qp3:eqn}).
Then $T_1$ and $T_i$ ($i=2,4$) are regarded as 
the time evolution and B\"acklund 
transformations of $q$-P$_{\rm III}$, respectively.

In order to formulate the symmetrization to $q$-P$_{\rm II}$, 
it is crucial to introduce the transformation $R_1$ defined by
\begin{equation}
 R_1=\pi^2 s_1,
\end{equation}
which satisfies
\begin{equation}
 {R_1}^2=T_1.
\end{equation}
The actions of $R_1$ are given by 
\begin{align}
 &R_1:~(a_0,a_1,a_2,c)\mapsto(a_2a_0,{a_0}^{-1},a_1a_0,c),\\
 &R_1(f_0) = \frac{qc^2}{f_0f_1}~\frac{1+a_0f_0}{a_0+f_0},\quad
 R_1(f_1) = f_0,
\end{align}
which describe the zig-zag motion around the line 
$a_2=q^{1/2}$ on the parameter space.
However, if we put $a_2=q^{1/2}$, 
then $R_1$ becomes the translation on the line
$a_2=q^{1/2}$ with the step $q^{1/2}$ 
(see Figure \ref{fig:R1}). In fact, the actions
of $R_1$ are now given by
\begin{align}
 &R_1:~(a_0,a_1,c)\mapsto(q^{1/2}a_0,q^{-1/2}a_1,c),\\
 &R_1(f_0) = \frac{qc^2}{f_0f_1}~\frac{1+a_0f_0}{a_0+f_0},\quad
 R_1(f_1) = f_0.\label{qp2:eqn2}
\end{align}
Applying ${R_1}^k{T_4}^N$ on (\ref{qp2:eqn2}) and putting 
\begin{equation}
 f_{i,N}^{k}={R_1}^k{T_4}^N(f_i)\quad (i=0,1,2), 
\end{equation}
we have
\begin{equation}\label{qp2:eqn3}
 f_{0,N}^{k+1} = \frac{q^{2N+1}c^2}{f_{0,N}^{k} f_{0,N}^{k-1}}~
 \frac{1+a_0q^{k/2} f_{0,N}^{k}}{a_0q^{k/2} + f_{0,N}^{k}},
\end{equation}
which is equivalent to $q$-P$_{\rm II}$, (\ref{qp2:eqn}).
Then $R_1$ and $T_4$ are regarded as the time evolution and the B\"acklund
transformation of $q$-P$_{\rm II}$, respectively.
\begin{figure}[h]
\begin{center}
\includegraphics[width=0.4\textwidth]{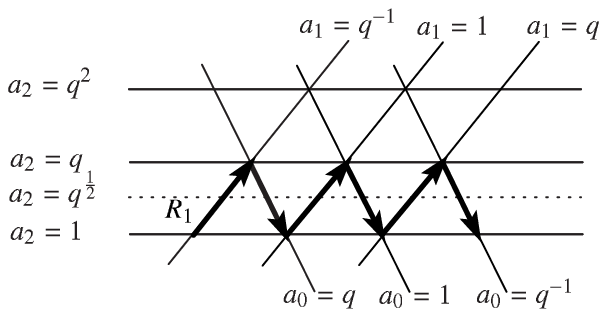}\hspace{3em}
\includegraphics[width=0.4\textwidth]{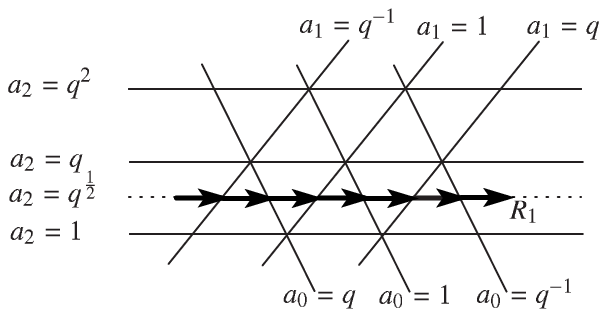}
\caption{Action of $R_1$ on the parameter space 
$\bm{a}=(a_0,a_1,a_2)\in(\mathbb{C}^{\times})^3$ with $a_0a_1a_2=q$.
Left: generic case. Right: $a_2=q^{1/2}$.}
\label{fig:R1}
\end{center}
\end{figure}

In general, it is possible to obtain various discrete dynamical 
systems of Painlev\'e type from
elements of infinite order that are not necessarily 
translations in the affine Weyl group by taking a projection on an appropriate
sublattice of corresponding root lattice. 
We call such a procedure a {\it projective reduction}.

By using the above formulation, we can now explain why the difference of
hypergeometric solutions to $q$-P$_{\rm III}$ and that to $q$-P$_{\rm II}$ occurs.
\subsection{Hypergeometric functions}
First, we explain about the difference of hypergeometric functions.
For convenience, we define the function $H_{n,m}$ by 
\begin{equation}
 H_{n,m}=\Theta({a_0}^2q^{2n+1};q^2)F_{n+\frac{1}{2},m},
\end{equation}
where $F_{n,m}$ is given in Remark \ref{remark:qp3_hyper}.
Then, we obtain from (\ref{qp3:3-term}) with 
$a_2=q^{1/2}$ the three-term relation for $H_{n,0}$:
\begin{equation}\label{qp3:3-term_2}
\begin{split}
 &H_{n+2,0}+\left({a_0}^{-2}q^{-2n-3}
 +{a_0}^{-2}q^{-2n-2}-1\right)H_{n+1,0}
 +{a_0}^{-4}q^{-4n-3}H_{n,0}\\
 &=\left[{T_1}^{n+2}+\left({a_0}^{-2}q^{-2n-3}
 +{a_0}^{-2}q^{-2n-2}-1\right){T_1}^{n+1}
 +{a_0}^{-4}q^{-4n-3}{T_1}^{n}\right]H_{0,0}=0.
\end{split}
\end{equation}
Since ${R_1}^2=T_1$, (\ref{qp3:3-term_2}) is
a fourth order difference equation  for 
$H_{n,0}={T_1}^n(H_{0,0})$ with respect to $R_1$.
Moreover, it admits the following factorization into two
linear difference operators:
\begin{equation}\label{qp2:factorization}
\begin{split}
 &{T_1}^{n+2}
 +\left({a_0}^{-2}q^{-2n-3}+{a_0}^{-2}q^{-2n-2}-1\right){T_1}^{n+1}
 +q^{-3-4n}{a_0}^{-4}{T_1}^{n}\\
 &=\left({R_1}^{n+3}+{R_1}^{n+2}+{a_0}^{-2}q^{-2n-2}{R_1}^{n+1}\right)
 \left({R_1}^{n+1}-{R_1}^{n}+{a_0}^{-2}q^{-n}{R_1}^{n-1}\right).
\end{split}
\end{equation}
On the other hand, the three-term relation for $G_n={R_1}^n(G_0)$ 
(see (\ref{qp2:3-term})) can be expressed as
\begin{equation}\label{qp2:3-term_2}
 \left({R_1}^{n+1}-{R_1}^n+{a_0}^{-2}q^{-n}{R_1}^{n-1}\right)G_0=0.
\end{equation}
Note that the second factor in the right-hand side of 
(\ref{qp2:factorization}) is exactly the operator in
(\ref{qp2:3-term_2}), thus, $G_n$ also satisfies (\ref{qp3:3-term_2}).
\subsection{Determinant structure}
Next, in order to discuss the difference of determinant structures, 
we need to introduce the $\tau$
functions and lift the representation to the Weyl group on the level of $\tau$
functions\cite{KNY:qp4,Tsuda:tau_qp34}. 
We introduce the new variables $\tau_i$ and
$\overline{\tau}_i$ ($i\in\mathbb{Z}/3\mathbb{Z}$) with
\begin{equation}
 f_i=q^{1/3}c^{2/3}
 \cfrac{\overline{\tau}_{i+1}\tau_{i-1}}{\tau_{i+1}\overline{\tau}_{i-1}}.
\end{equation}
\begin{proposition}[\cite{Tsuda:tau_qp34}]\label{prop:action_tau}
We define the action of $s_i$ $(i=0,1,2)$, 
$\pi$, $w_j$ $(j=0,1)$, and $r$ on $\tau_k$ and
$\overline{\tau}_k$ $(k=0,1,2)$ by the following formulae$:$
\begin{equation}
 \left\{\begin{array}{l}\medskip
 {\displaystyle 
 s_i(\tau_i)=
  \frac{u_i\tau_{i+1}\overline{\tau}_{i-1}+\overline{\tau}_{i+1}\tau_{i-1}}
  {{u_i}^{1/2}\overline{\tau}_i},\quad
 s_i(\tau_j)=\tau_j\quad (i\neq j),}\\
 {\displaystyle 
 s_i(\overline{\tau}_i)=
  \frac{v_i\overline{\tau}_{i+1}\tau_{i-1}+\tau_{i+1}\overline{\tau}_{i-1}}
  {{v_i}^{1/2}\tau_i},\quad
 s_i(\overline{\tau}_j)=\overline{\tau}_j\quad (i\neq j),}
 \end{array}\right.
\end{equation}
\begin{equation}
 \pi(\tau_i)=\tau_{i+1},\quad 
 \pi(\overline{\tau}_i)=\overline{\tau}_{i+1},
\end{equation}
\begin{equation}
\left\{
\begin{array}{l}\medskip
 {\displaystyle
 w_0(\overline{\tau}_i)= 
  \frac{{a_{i+1}}^{1/3}(\overline{\tau}_i\tau_{i+1}\tau_{i+2}
  +u_{i-1}\tau_i\overline{\tau}_{i+1}\tau_{i+2}
  +{u_{i+1}}^{-1}\tau_i\tau_{i+1}\overline{\tau}_{i+2})}
  {{a_{i+2}}^{1/3}\overline{\tau}_{i+1}\overline{\tau}_{i+2}},}\\
 {\displaystyle w_0(\tau_i) = \tau_i,}
\end{array}\right.
\end{equation}
\begin{equation}
 \left\{\begin{array}{l}\medskip
 {\displaystyle 
 w_1(\tau_i) =
  \frac{{a_{i+1}}^{1/3}(\tau_i\overline{\tau}_{i+1}\overline{\tau}_{i+2}
   + v_{i-1}\overline{\tau}_i\tau_{i+1}\overline{\tau}_{i+2}
   + {v_{i+1}}^{-1}\overline{\tau}_i\overline{\tau}_{i+1}\tau_{i+2})}
  {{a_{i+2}}^{1/3}\tau_{i+1}\tau_{i+2}},}\\
 {\displaystyle w_1(\overline{\tau}_i) = \overline{\tau}_i,}
\end{array}\right.
\end{equation}
\begin{equation}
 r(\tau_i) = \overline{\tau}_i,\quad
 r(\overline{\tau}_i) = \tau_i,
\end{equation}
with
\begin{equation}
 u_i = q^{-1/3}c^{-2/3}a_i,\quad
 v_i = q^{1/3}c^{2/3}a_i,
\end{equation}
where $i,j\in\mathbb{Z}/3\mathbb{Z}$.
Then, $\langle s_0,s_1,s_2,\pi,w_0,w_1,r\rangle$ realizes the affine Weyl group
$\widetilde{W}((A_2+A_1)^{(1)})$.
\end{proposition}
\begin{figure}[h]
\begin{center}
\includegraphics[width=0.6\textwidth]{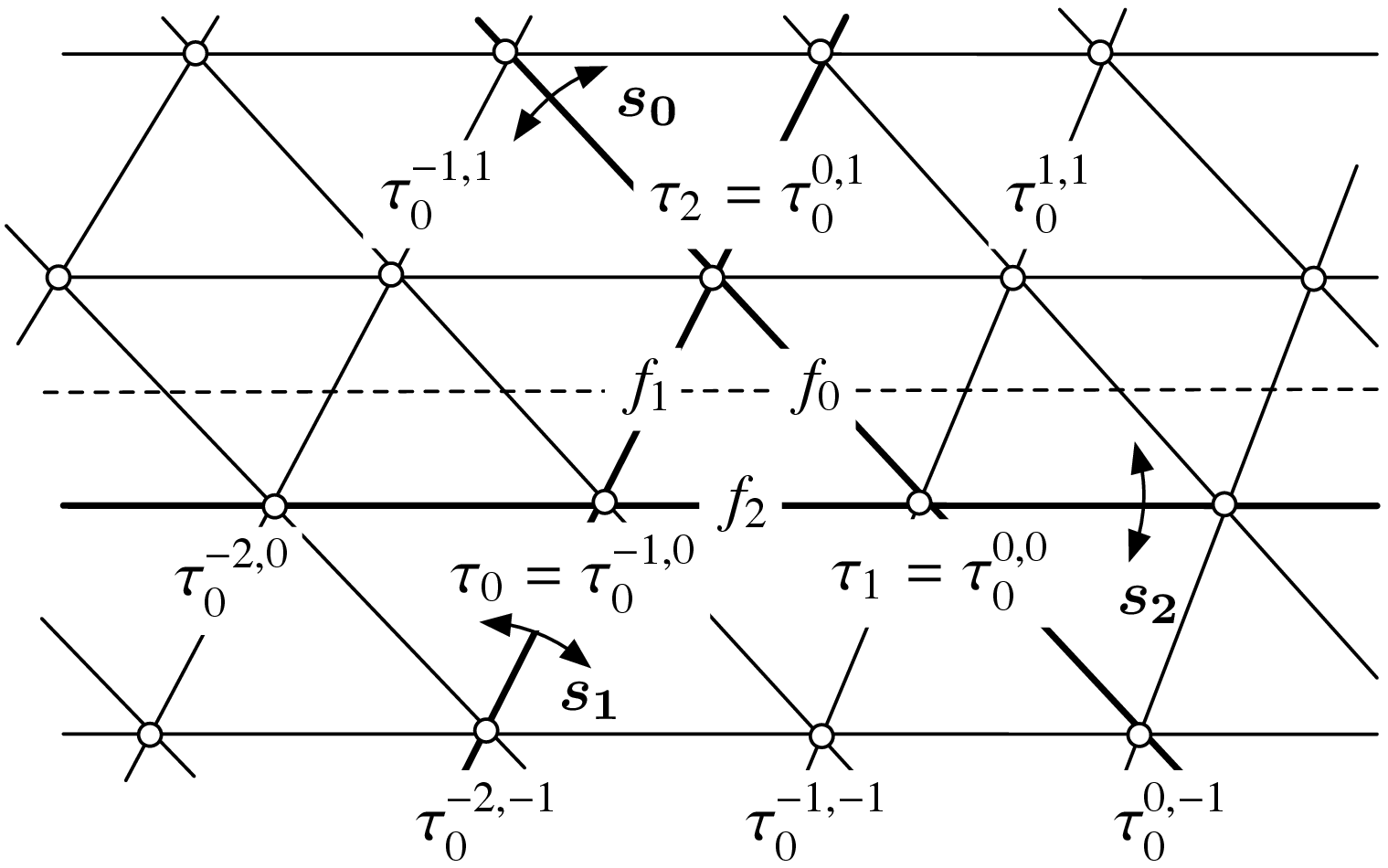} \quad
\raise40pt\hbox{\includegraphics[width=0.2\textwidth]{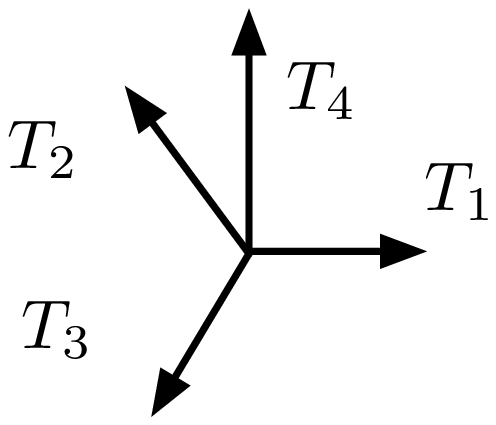}}
\end{center}
\caption{Configuration of the $\tau$ functions on the lattice with $N=0$.}
\end{figure}
Then, we define the $\tau$ functions $\tau^{n,m}_{N}$ ($n, m, N \in\bm{Z}$) by
\begin{equation}\label{notation:tau}
 \tau^{n,m}_{N}={T_1}^n{T_2}^m{T_4}^N(\tau_1).
\end{equation}
We note that $\tau_0=\tau^{-1,0}_{0}$,
$\tau_1=\tau^{0,0}_{0}$, $\tau_2=\tau^{0,1}_{0}$, $\overline{\tau}_0=\tau^{-1,0}_{1}$,
$\overline{\tau}_1=\tau^{0,0}_{1}$, and $\overline{\tau}_2=\tau^{0,1}_{1}$.
\begin{proposition}
The action of $\widetilde{W}((A_2+A_1)^{(1)})$ on $\tau^{n,m}_{N}$ is
\begin{align}
 &s_0(\tau^{n,m}_{N})=\tau^{-n,m-n}_{N},\quad
 s_1(\tau^{n,m}_{N})=\tau^{m-1,n+1}_{N},\quad
 s_2(\tau^{n,m}_{N})=\tau^{n-m,-m}_{N},\quad
 \pi(\tau^{n,m}_{N})=\tau^{-m,n-m+1}_{N},\\
 &w_0(\tau^{n,m}_{N})=\tau^{n,m}_{-N},\quad
 w_1(\tau^{n,m}_{N})=\tau^{n,m}_{2-N},\quad
 r(\tau^{n,m}_{N})=\tau^{n,m}_{1-N}.
\end{align}
\end{proposition}
For convenience, we put
\begin{equation}\label{notation:parameters}
 \alpha_i={a_i}^{1/6},\quad
 \gamma=c^{1/6},\quad Q=q^{1/6}.
\end{equation}
Though it is possible to derive more various 
bilinear difference equations from Proposition
\ref{prop:action_tau}, we present here only the equations 
that are directly relevant to $q$-P$_{\rm III}$, (\ref{qp3:eqn2}).
\begin{proposition}\label{prop:qp3_bl}
 The following bilinear equations hold$:$
\begin{align}
 &\tau^{n,m}_{N+1}\tau^{n+1,m+1}_{N}
  -Q^{-3n+3m+2N-2}\gamma^{2}{\alpha_1}^{3}\tau^{n+1,m}_{N}\tau^{n,m+1}_{N+1}
  +Q^{-6n+6m+4N-4}\gamma^{4}{\alpha_1}^{6}\tau^{n,m}_{N}\tau^{n+1,m+1}_{N+1}=0,
 \label{qp3_bl:typeVII_1}\\
 &\tau^{n+1,m+1}_{N+1}\tau^{n+1,m}_{N}
  -Q^{3n+2N+4}\gamma^{2}{\alpha_0}^{3}\tau^{n,m}_{N}\tau^{n+2,m+1}_{N+1}
  +Q^{6n+4N+8}\gamma^{4}{\alpha_0}^{6}\tau^{n+1,m+1}_{N}\tau^{n+1,m}_{N+1}=0,
 \label{qp3_bl:typeVII_3}\\
 &\tau^{n+1,m+1}_{N+1}\tau^{n,m}_{N}
  -Q^{-3n+3m-2N-4}\gamma^{-2}{\alpha_1}^{3}\tau^{n+1,m}_{N+1}\tau^{n,m+1}_{N}
  +Q^{-6n+6m-4N-8}\gamma^{-4}{\alpha_1}^{6}\tau^{n+1,m+1}_{N}\tau^{n,m}_{N+1}=0,
 \label{qp3_bl:typeVIII_1}\\
 &\tau^{n+1,m}_{N+1}\tau^{n+1,m+1}_{N}
  -Q^{3n-2N+2}\gamma^{-2}{\alpha_0}^{3}\tau^{n,m}_{N+1}\tau^{n+2,m+1}_{N}
  +Q^{6n-4N+4}\gamma^{-4}{\alpha_0}^{6}\tau^{n+1,m}_{N}\tau^{n+1,m+1}_{N+1}=0,
 \label{qp3_bl:typeVIII_3}\\
 &\tau^{n,m}_{N+1}\tau^{n,m}_{N-1}
  +Q^{-8n+4m-4}{\alpha_0}^{-4}{\alpha_1}^{4}\left(\tau^{n,m}_{N}\right)^2
  -Q^{-2n+m-1}{\alpha_0}^{-1}\alpha_1\tau^{n+1,m}_{N}\tau^{n-1,m}_{N}=0.
 \label{qp3_bl:typeI_3}
\end{align}
\end{proposition}
The proof of Proposition \ref{prop:qp3_bl} will be given in Appendix \ref{sec:bl_qp3}. 

As seen below $q$-P$_{\rm III}$, (\ref{qp3:eqn2}) 
or (\ref{qp3:eqn3}), can be obtained from the
bilinear equations. Noticing that
\begin{equation}
 f_{0,N}^{n,m}=Q^{4N+2}\gamma^4
  \frac{\tau^{n,m}_{N+1}\tau^{n,m+1}_{N}}{\tau_N^{n,m}\tau^{n,m+1}_{N+1}},\quad
 f_{1,N}^{n,m}=Q^{4N+2}\gamma^4
  \frac{\tau^{n,m+1}_{N+1}\tau^{n-1,m}_{N}}{\tau_N^{n,m+1}\tau^{n-1,m}_{N+1}},\quad
 f_{2,N}^{n,m}=Q^{4N+2}\gamma^4
  \frac{\tau^{n-1,m}_{N+1}\tau^{n,m}_{N}}{\tau_N^{n-1,m}\tau^{n,m}_{N+1}},
\end{equation}
we can rewrite (\ref{qp3_bl:typeVII_1}) and (\ref{qp3_bl:typeVIII_1}) as
\begin{align}
 &1 + Q^{-6n+6m-6}{\alpha_1}^6 f_{1,N}^{n+1,m}=Q^{-3n+3m+2N-2}\gamma^2{\alpha_1}^3
  \frac{\tau^{n+1,m}_{N}\tau^{n,m+1}_{N+1}}{\tau^{n,m}_{N+1}\tau^{n+1,m+1}_{N}},
  \label{qp3:derivation1}\\
 &1 + Q^{6n-6m+6}{\alpha_1}^{-6}f_{1,N}^{n+1,m}=Q^{3n-3m+2N+4}\gamma^2{\alpha_1}^{-3}
  \frac{\tau^{n+1,m}_{N+1}\tau^{n,m+1}_{N}}{\tau^{n,m}_{N+1}\tau^{n+1,m+1}_{N}},
  \label{qp3:derivation2}
\end{align}
respectively. Dividing (\ref{qp3:derivation2}) by (\ref{qp3:derivation1}), we have
\begin{equation}
 \frac{1 + Q^{6n-6m+6}{\alpha_1}^{-6}f_{1,N}^{n+1,m}}
  {1 + Q^{-6n+6m-6}{\alpha_1}^6 f_{1,N}^{n+1,m}}
  =Q^{6n-6m+6}{\alpha_1}^{-6}\frac{\tau^{n+1,m}_{N+1}\tau^{n,m+1}_{N}}
  {\tau^{n+1,m}_{N}\tau^{n,m+1}_{N+1}}
  =Q^{6n-6m+6}{\alpha_1}^{-6}\frac{f^{n,m}_{0,N}}{f^{n+1,m}_{2,N}},
\end{equation}
which is equivalent to the second equation of (\ref{qp3:eqn3}). Similarly,
(\ref{qp3_bl:typeVII_3}) and (\ref{qp3_bl:typeVIII_3}) yield the first equation of
(\ref{qp3:eqn3}).

For the hypergeometric solutions, we relate the $\tau$ functions to the
determinants $\psi^{n,m}_{N}$, (\ref{qp3:det}), 
by multiplication of appropriate ``gauge'' factor. 
We set
\begin{align}\label{tau:gauge}
 \tau^{n,m}_{N} = &(-1)^{N(N+1)/2}
 Q^{-2(2n-m)N^2+6Nn}
 {\alpha_0}^{-4N^2+6N}
 {\alpha_2}^{-2N^2}
 \left(
 \frac{\Theta(-Q^{-6n}{\alpha_0}^{-6};Q^6)
  \Theta(-Q^{6m}{\alpha_2}^{-6};Q^6)}
  {\Theta(Q^{-6(n-m)}{\alpha_0}^{-6}{\alpha_2}^{-6};Q^6)}
 \right)^N\nonumber\\
 &\times
 \Gamma(Q^{2n-m+1}{\alpha_0}^{2}\alpha_2;Q,Q)
 \Gamma(Q^{-n+2m-1}{\alpha_1}^{2}\alpha_0;Q,Q)
 \Gamma(Q^{-n-m}{\alpha_2}^{2}\alpha_1;Q,Q)
 ~\psi^{n,m-1}_{N},
\end{align}
and put $\gamma=1$.
Then the bilinear equations (\ref{qp3_bl:typeVII_1})--(\ref{qp3_bl:typeI_3}) 
can be rewritten as
\begin{align}
 &\psi^{n,m}_{N+1}\psi^{n+1,m+1}_{N}
  -Q^{-12n+12N}{\alpha_0}^{-12}\psi^{n+1,m}_{N}\psi^{n,m+1}_{N+1}
  +Q^{-12n}{\alpha_0}^{-12}\psi^{n,m}_{N}\psi^{n+1,m+1}_{N+1}=0,
 \label{qp3_bl:psi1}\\
 &\psi^{n+1,m+1}_{N+1}\psi^{n+1,m}_{N}
  -Q^{-12N}\psi^{n,m}_{N}\psi^{n+2,m+1}_{N+1}
  -Q^{12n+12}{\alpha_0}^{12}\psi^{n+1,m+1}_{N}\psi^{n+1,m}_{N+1}=0,
 \label{qp3_bl:psi2}\\
 &\psi^{n+1,m+1}_{N+1}\psi^{n,m}_{N}
  -\psi^{n+1,m}_{N+1}\psi^{n,m+1}_{N}
  +Q^{12m+12}{\alpha_2}^{-12}\psi^{n+1,m+1}_{N}\psi^{n,m}_{N+1}=0,
 \label{qp3_bl:psi3}\\
 &\psi^{n+1,m}_{N+1}\psi^{n+1,m+1}_{N}
  -Q^{12m+12}{\alpha_2}^{-12}\psi^{n,m}_{N+1}\psi^{n+2,m+1}_{N}
  -\psi^{n+1,m}_{N}\psi^{n+1,m+1}_{N+1}=0,
 \label{qp3_bl:psi4}\\
 &\psi^{n,m}_{N+1}\psi^{n,m}_{N-1}
  -\left(\psi^{n,m}_{N}\right)^2
  +\psi^{n+1,m}_{N}\psi^{n-1,m}_{N}=0,
 \label{qp3_bl:psi5}
\end{align}
respectively. Equations (\ref{qp3_bl:psi1})--(\ref{qp3_bl:psi4}) are equivalent to
(\ref{qp3:bl1})--(\ref{qp3:bl4}). Note that (\ref{qp3_bl:psi5}) is exactly the discrete
Toda equation, (\ref{dToda:bl}), which fixes the determinant structure of the hypergeometric
solutions as mentioned in Remark \ref{remark:qp3_hyper}.
\begin{remark}\label{rem:hyper_tau}\rm
The gauge factor $\tau^{n,m}_{N}/\psi^{n,m-1}_{N}$ in (\ref{tau:gauge}) 
is obtained by solving the overdetermined
system of the bilinear difference equations with $\gamma=1$ under the boundary conditions
$\tau^{n,m}_{N}=0$ $(N\in\mathbb{Z}_{<0})$\cite{Nakazono:qp3_gauge}.
\end{remark}

Let us consider the bilinear equations for $q$-P$_{\rm II}$.
Since we need $R_1$, $\tau_i$, and $\overline{\tau}_i$ $(i\in\mathbb{Z}/3\mathbb{Z})$,
the lattice is restricted to the ``unit-strip'' (see Figure \ref{fig:R1_strip}).
Therefore, we have only to consider $\tau^{n,0}_{N}$ and $\tau^{n,1}_{N}$ ($n, N\in\mathbb{Z}$).
We set 
\begin{equation}\label{notation:tau_qp2}
 \tau^k_N = {R_1}^k{T_4}^N(\tau_1).
\end{equation}
Note that 
\begin{equation}\label{notation:tau_qp2_2}
 \tau_0 = \tau^{-2}_0,\quad
 \tau_1 = \tau^{0}_0,\quad
 \tau_2 = \tau^{-1}_0,\quad
 \overline{\tau}_0 = \tau^{-2}_1,\quad
 \overline{\tau}_1 = \tau^{0}_1,\quad
 \overline{\tau}_2 = \tau^{-1}_1.
\end{equation}
In general, it follows that
\begin{equation}\label{notation:tau_qp2_qp3}
 \tau^{n,0}_N = \tau^{2n}_{N},\quad  \tau^{n,1}_N = \tau^{2n-1}_{N},
\end{equation}
and
\begin{equation}
 f_{0,N}^{k}=Q^{4N+2}\gamma^4 \frac{\tau^{k}_{N+1}\tau^{k-1}_N}{\tau^k_N\tau^{k-1}_{N+1}}.
\end{equation}
\begin{figure}[h]
\begin{center}
\includegraphics[width=0.5\textwidth]{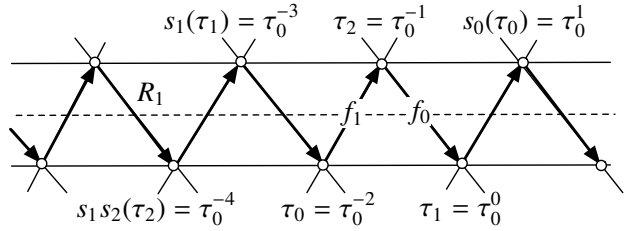} 
\end{center}
\caption{The actions of $R_1$ on $\tau_i$ $(i=0,1,2)$.}
\label{fig:R1_strip}
\end{figure}
\begin{proposition}\label{prop:qp2_bl}
The following bilinear equations hold$:$
\begin{align}
 &Q^{(-3k+4N+2)/2}\gamma^{2}{\alpha_0}^{-3}\tau^{k+1}_{N}\tau^{k-2}_{N+1}
  -Q^{-3k+4N+2}\gamma^{4}{\alpha_0}^{-6}\tau^{k-1}_{N}\tau^{k}_{N+1}
  -\tau^{k-1}_{N+1}\tau^{k}_{N}=0,
 \label{qp2_bl:Prop_A2_3}\\
 &Q^{(-3k-4N-2)/2}\gamma^{-2}{\alpha_0}^{-3}\tau^{k+1}_{N+1}\tau^{k-2}_{N}
  -Q^{-3k-4N-2}\gamma^{-4}{\alpha_0}^{-6}\tau^{k}_{N}\tau^{k-1}_{N+1}
  -\tau^{k}_{N+1}\tau^{k-1}_{N}=0,
 \label{qp2_bl:Prop_A2_4}\\
 &\tau^{k}_{N+1}\tau^{k+1}_{N-1}
  -Q^{(k-4N+1)/2}\gamma^{-2}\alpha_0\tau^{k+2}_{N}\tau^{k-1}_{N}
  -Q^{-k+4N-1}\gamma^{4}{\alpha_0}^{-2}\tau^{k}_{N}\tau^{k+1}_{N}=0.
 \label{qp2_bl:Prop_A2_1}
\end{align}
\end{proposition}
The proof of Proposition \ref{prop:qp2_bl} 
will be given in Appendix \ref{sec:bl_qp2}. 

One can obtain $q$-P$_{\rm II}$, (\ref{qp2:eqn2}), 
from Proposition \ref{prop:qp2_bl} as
follows. Equations (\ref{qp2_bl:Prop_A2_3}) and 
(\ref{qp2_bl:Prop_A2_4}) can be rewritten as
\begin{align}
 &1+Q^{-3k}{\alpha_0}^{-6}f^k_{0,N}
  =Q^{(-3k+4N+2)/2}\gamma^{2}{\alpha_0}^{-3}
   \frac{\tau^{k+1}_{N}\tau^{k-2}_{N+1}}{\tau^{k-1}_{N+1}\tau^{k}_{N}},
 \label{qp2:derivation1}\\
 &1+Q^{3k}{\alpha_0}^6 f^k_{0,N}
  =Q^{(3k+4N+2)/2}\gamma^{2}{\alpha_0}^{3}
   \frac{\tau^{k+1}_{N+1}\tau^{k-2}_{N}}{\tau^{k-1}_{N+1}\tau^{k}_{N}}.
 \label{qp2:derivation2}
\end{align}
Dividing (\ref{qp2:derivation2}) by (\ref{qp2:derivation1}), we have
\begin{equation}
 \frac{1+Q^{3k}{\alpha_0}^6 f^k_{0,N}}{1+Q^{-3k}{\alpha_0}^{-6}f^k_{0,N}}
  =Q^{3k}{\alpha_0}^6\frac{\tau^{k+1}_{N+1}\tau^{k-2}_{N}}{\tau^{k+1}_{N}\tau^{k-2}_{N+1}}
  =Q^{3k-12N-6}\gamma^{12}{\alpha_0}^6 f_{0,N}^{k+1}f_{0,N}^{k}f_{0,N}^{k-1},
\end{equation}
which is equivalent to (\ref{qp2:eqn3}).

For hypergeometric solutions, by putting $\gamma=1$ and
\begin{align}
 \tau^{k}_{N} = (-1)^{N(N-1)/2}Q^{N(N-1)(k+n)}{\alpha_0}^{2N(N-1)}
  \frac{\Gamma(Q^{(2k+3)/2}{\alpha_0}^{2};Q,Q)
  \Gamma(Q^{-k/2}{\alpha_0}^{-1};Q,Q)
  \Gamma(Q^{(-k+3)/2}{\alpha_0}^{-1};Q,Q)}
  {\Theta(Q^{3k+1}{\alpha_0}^{6};Q^3)^N}
  ~\phi^k_{N}, \label{qp2_tau:gauge}
\end{align}
we can rewrite the bilinear equations (\ref{qp2_bl:Prop_A2_3}),
(\ref{qp2_bl:Prop_A2_4}), and (\ref{qp2_bl:Prop_A2_1}) as
\begin{align}
 &Q^{6N-6k+6}{\alpha_0}^{-12}\phi^{k+1}_{N}\phi^{k-2}_{N+1}
 +Q^{6N}\phi^{k-1}_{N}\phi^{k}_{N+1}
 -\phi^{k-1}_{N+1}\phi^{k}_{N}=0, \label{qp2_bl:phi1}\\
 &Q^{6N}\phi^{k+1}_{N+1}\phi^{k-2}_{N}
 +Q^{-6N-6k}{\alpha_0}^{-12}\phi^{k}_{N}\phi^{k-1}_{N+1}
 -\phi^{k}_{N+1}\phi^{k-1}_{N}=0,\label{qp2_bl:phi2}\\
 &\phi^{k}_{N+1}\phi^{k+1}_{N-1}
 -\phi^{k}_{N}\phi^{k+1}_{N}
 +\phi^{k+2}_{N}\phi^{k-1}_{N}=0,\label{qp2_bl:phi3}
\end{align}
which are equivalent to (\ref{qp2:bl1}), (\ref{qp2:bl2}), and
(\ref{dToda2:bl}), respectively. 
The determinant structure of the hypergeometric solutions is
fixed by (\ref{qp2_bl:phi3}) as was explained in Remark \ref{remark:qp2_hyper}.

Therefore, the difference of the determinant structures of the hypergeometric solutions to
$q$-P$_{\rm III}$ and that to $q$-P$_{\rm II}$ originates from the following procedures:
\begin{description}
 \item[{\rm (i)}] 
  the specialization $a_2=q^{1/2}$
  and the restriction of $\tau$ functions on the ``unit-strip'';
 \item[{\rm (ii)}] taking the half-step translation $R_1$ instead of $T_1$ as a time evolution.
\end{description}
These result in the difference of the bilinear equations
(\ref{qp3_bl:typeI_3}) (or (\ref{qp3_bl:psi5})) and 
(\ref{qp2_bl:Prop_A2_1}) (or (\ref{qp2_bl:phi3})),
which fix the determinant structure of the hypergeometric solutions.
\section{Concluding remarks}
In this paper, we have clarified the mechanism 
that gives rise to the apparent ``inconsistency'' in
the hypergeometric solutions to $q$-P$_{\rm III}$ and that to $q$-P$_{\rm II}$ by using their
underlying affine Weyl group symmetry. In general, it is
also possible to explain the inconsistency 
among the hypergeometric solutions to other symmetric
and asymmetric discrete Painlev\'e equations 
(see, for example, Propositions \ref{prop:sdP2} and
\ref{prop:adP2}).

Before closing, we demonstrate another example of the projective reductions.
Let us consider the following system of difference equations\cite{Ohta:RIMS_dP}:
\begin{equation}\label{ternarized_dP1:eqn}
 Z_n + X_n = \frac{3na+b_1}{Y_n} + t,\quad
 X_{n+1} + Y_n = \frac{(3n+1)a+b_2}{Z_n} + t,\quad 
 Y_{n+1} + Z_n = \frac{(3n+2)a+b_3}{X_{n+1}} + t,
\end{equation}
where $X_n$, $Y_n$, and $Z_n$ are the dependent variables, 
$n\in\mathbb{Z}$ is the independent variable, and
$a,b_1,b_2,b_3,t\in\mathbb{C}$ are parameters. 
Equation (\ref{ternarized_dP1:eqn}) is one of
the discrete Painlev\'e systems of type $A_3^{(1)}$. 
Namely, it arises from a B\"acklund
transformation of the Painlev\'e V equation, 
which describes a translation in a different direction
from (\ref{adP2:eqn}). Putting $b_1=b_2=b_3=b$, $X_n=x_{3n-1}$, $Y_n=x_{3n}$, and
$Z_n=x_{3n+1}$, we can reduce (\ref{ternarized_dP1:eqn}) to
\begin{equation}\label{dP1:eqn}
 x_{n+1}+x_{n-1}=\frac{an+b}{x_n}+t,
\end{equation}
which is known as a discrete Painlev\'e I equation\cite{RG:coales}.
This reduction from (\ref{ternarized_dP1:eqn}) to (\ref{dP1:eqn}) is a typical
example of the projective reductions other than a symmetrization. 

It seems that various projective reductions of the discrete Painlev\'e systems
change the underlying symmetry and yield a number of intriguing problems. 
One interesting project is to make a list of
the hypergeometric functions that appear as the solutions to all the symmetric
discrete Painlev\'e equations in Sakai's
classification\cite{Sakai:Painleve,KMNOY:hyper1,KMNOY:hyper2}.
These will be discussed in forthcoming papers\cite{KN:hyper}.
\par\bigskip
\noindent{\bf Acknowledgement.} 
The authors would like to express their sincere thanks to
Prof. M. Noumi for fruitful discussions and valuable suggestions. 
They acknowledge continuous
encouragement by Prof. T. Masuda, Prof. H. Sakai, and Prof. Y. Yamada. 
They also appreciate the valuable comments from the referees
which have improved the quality of this paper.
This work has been partially
supported by the JSPS Grant-in-Aid for Scientific Research No. 19340039.

\appendix
\renewcommand{\theequation}{A.\arabic{equation}}
\section{On the difference equation (\ref{qp2:3-term})}
\label{appendix:A}
In this appendix, we consider the 
equation (\ref{qp2:3-term}) (or (\ref{qp2:3-term2})):
\begin{equation}\label{qp2:3-term3}
 U_{k+1}-U_k+\frac{1}{{a_0}^2q^{k}}U_{k-1}=0.
\end{equation}
Recall that we have obtained two solutions to the equation above,
i.e, $G_k$ (\ref{qp2:Gk}) and $\hat{G}_k$ (\ref{eqn:hat_Gk}).
These are described as follows:
\begin{align}
 &G_k=A_kv_k+B_kw_k,\label{sol:Gk}\\
 &v_k=\Theta(ia_0q^{(2k+1)/4};q^{1/2})
  {}_1\varphi_1\left(\begin{matrix}0\\-q^{1/2}\end{matrix}
  ;q^{1/2},-ia_0q^{(3+2k)/4}\right),\\
 &w_k=\Theta(-ia_0q^{(2k+1)/4};q^{1/2})
  {}_1\varphi_1\left(\begin{matrix}0\\-q^{1/2}\end{matrix}
  ;q^{1/2},ia_0q^{(3+2k)/4}\right),
\end{align}
and
\begin{align}
 &\hat{G}_k=\hat{A}_k\hat{v}_k+\hat{A}_{k+1}\hat{w}_k,\label{sol:hatGk}\\
 &\hat{v}_k=\cfrac{\Theta({a_0}^2q^k;q^2)}{(q^{-1};q^2)_\infty}
  {}_1\varphi_1\left(\begin{matrix}0\\ q^3\end{matrix}
  ;q^{2},{a_0}^2q^{k+3}\right),\\
 &\hat{w}_k=\cfrac{\Theta({a_0}^2q^{k+1};q^2)}{(q;q^2)_\infty}
  {}_1\varphi_1\left(\begin{matrix}0\\ q\end{matrix}
  ;q^{2},{a_0}^2q^{k+2}\right).
\end{align}
Here $A_k$ and $B_k$ are periodic functions of period one, 
and $\hat{A}_k$ is that of period two.
For an initial value $(U_0,U_1)$ given,
the values of $U_k$ ($k\in\mathbb{Z}$) are 
determined recursively by (\ref{qp2:3-term3}).
Since the Casoratians 
$\begin{vmatrix}v_0&w_0\\ v_1&w_1\end{vmatrix}$ and 
$\begin{vmatrix}\hat{v}_0&\hat{w}_0\\ \hat{w}_1&\hat{v}_1\end{vmatrix}$ 
do not vanish for generic values of $a_0$ and $q$,
the coefficients of (\ref{sol:Gk}) and (\ref{sol:hatGk}) 
are specified by the initial value as
\begin{align}
 &\begin{pmatrix}A_0\\ B_0\end{pmatrix}
  =\begin{pmatrix}v_0&w_0\\ v_1&w_1\end{pmatrix}^{-1}
  \begin{pmatrix}U_0\\ U_1\end{pmatrix},
 \label{eqn:A_B_initial}\\
 &\begin{pmatrix}\hat{A}_0\\ \hat{A}_1\end{pmatrix}
  =\begin{pmatrix}\hat{v}_0&\hat{w}_0\\ \hat{w}_1&\hat{v}_1\end{pmatrix}^{-1}
  \begin{pmatrix}U_0\\ U_1\end{pmatrix}.
 \label{eqn:hatA_initial}
\end{align}
Hence we conclude that (\ref{sol:Gk}) and (\ref{sol:hatGk})
give two different expressions of the general solution to (\ref{qp2:3-term3}).

Next we shall show an identity among the basic hypergeometric series
${}_1\varphi_1$ with two different bases $q^2$ and $q^{1/2}$.
It follows from (\ref{eqn:A_B_initial}) and (\ref{eqn:hatA_initial}) that
\begin{equation}\label{eqn:A_B_hatA}
 \begin{pmatrix}\hat{A}_0\\ \hat{A}_1\end{pmatrix}
 =\begin{pmatrix}
  \cfrac{v_0\hat{v}_1-v_1\hat{w}_0}{\hat{v}_0\hat{v}_1-\hat{w}_0\hat{w}_1}
  &\cfrac{w_0\hat{v}_1-w_1\hat{w}_0}{\hat{v}_0\hat{v}_1-\hat{w}_0\hat{w}_1}\\
  \cfrac{v_1\hat{v}_0-v_0\hat{w}_1}{\hat{v}_0\hat{v}_1-\hat{w}_0\hat{w}_1}
  &\cfrac{w_1\hat{v}_0-w_0\hat{w}_1}{\hat{v}_0\hat{v}_1-\hat{w}_0\hat{w}_1}
 \end{pmatrix}
 \begin{pmatrix}A_0\\ B_0\end{pmatrix}.
\end{equation}
By definition we can express 
$v_k$, $w_k$, $\hat{v}_k$, and $\hat{w}_k$ as functions in $a_0$, namely, 
\begin{equation}
 v_k=v(a_0q^{k/2}),\quad w_k=w(a_0q^{k/2}),\quad
 \hat{v}_k=\hat{v}(a_0q^{k/2}),\quad \hat{w}_k=\hat{w}(a_0q^{k/2}).
\end{equation}
Note that $w(a_0)=v(-a_0)$.
Substituting $A_0=0$ (or $B_0=0$) in (\ref{eqn:A_B_hatA})
leads to the following formula:
\begin{equation}\label{eqn:formula_hypergeometric}
 v(a_0q^n)=y(a_0)\hat{v}(a_0q^n)+z(a_0)\hat{w}(a_0q^n),
\end{equation}
where
\begin{equation}
 y(a_0)=\cfrac{v(a_0)\hat{v}(a_0q^{1/2})-v(a_0q^{1/2})\hat{w}(a_0)}
  {\hat{v}(a_0)\hat{v}(a_0q^{1/2})-\hat{w}(a_0)\hat{w}(a_0q^{1/2})},\quad
 z(a_0)=\cfrac{v(a_0q^{1/2})\hat{v}(a_0)-v(a_0)\hat{w}(a_0q^{1/2})}
  {\hat{v}(a_0)\hat{v}(a_0q^{1/2})-\hat{w}(a_0)\hat{w}(a_0q^{1/2})}.
\end{equation}
Also, we have $z(a_0)=y(a_0q^{1/2})$ and $y(a_0)=z(a_0q^{1/2})$ from 
(\ref{eqn:formula_hypergeometric}) with $n=1$.
\renewcommand{\theequation}{B.\arabic{equation}}
\section{Derivation of bilinear equations}
In this appendix, we derive various bilinear equations 
for $\tau$ functions from the birational
representations of $\widetilde{W}((A_2+A_1)^{(1)})$ 
given in Proposition \ref{prop:action_tau}.
\subsection{Bilinear equations for ${\bm q}$-P$_{\rm\bf III}$}
\label{sec:bl_qp3}
We use the notations introduced in (\ref{notation:tau}) and (\ref{notation:parameters}).
For convenience, we classify the bilinear equations into six types
so that any equations which belong to the same type can be transformed into 
each other by the action of $\widetilde{W}((A_2+A_1)^{(1)})$.
\begin{proposition}[Type I: Discrete Toda type]
\label{prop:dToda}
The following bilinear equations hold$:$
\begin{align}
 &\tau^{n,m}_{N+1}\tau^{n,m}_{N-1}
  +Q^{4n-8m+4}{\alpha_1}^{-4}{\alpha_2}^{4}\left(\tau^{n,m}_{N}\right)^2
  -Q^{n-2m+1}{\alpha_1}^{-1}\alpha_2\tau^{n,m+1}_{N}\tau^{n,m-1}_{N} = 0,
 \label{TypeI_1}\\
 &\tau^{n,m}_{N+1}\tau^{n,m}_{N-1}
  +Q^{4n+4m}{\alpha_0}^{4}{\alpha_2}^{-4}\left(\tau^{n,m}_{N}\right)^2
  -Q^{n+m}\alpha_0{\alpha_2}^{-1}\tau^{n+1,m+1}_{N}\tau^{n-1,m-1}_{N} = 0,
 \label{TypeI_2}\\
 &\tau^{n,m}_{N+1}\tau^{n,m}_{N-1}
  +Q^{-8n+4m-4}{\alpha_0}^{-4}{\alpha_1}^{4}\left(\tau^{n,m}_{N}\right)^2
  -Q^{-2n+m-1}{\alpha_0}^{-1}\alpha_1\tau^{n+1,m}_{N}\tau^{n-1,m}_{N} = 0.
 \label{TypeI_3}
\end{align} 
\end{proposition}
\begin{figure}[h]
\begin{center}
~\raise3pt\hbox{\includegraphics[width=0.22\textwidth]{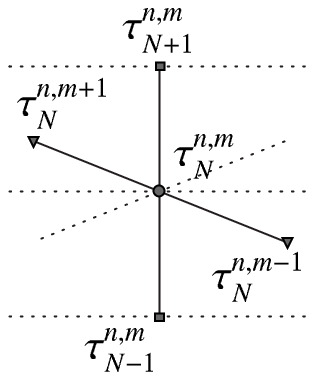}}\hspace{1em}
\raise3pt\hbox{\includegraphics[width=0.23\textwidth]{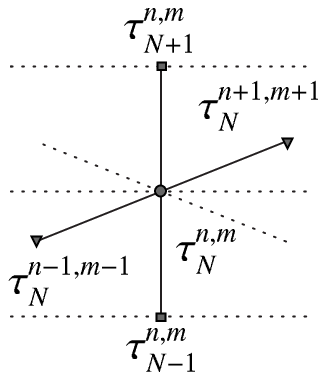}}\hspace{1em}
\includegraphics[width=0.24\textwidth]{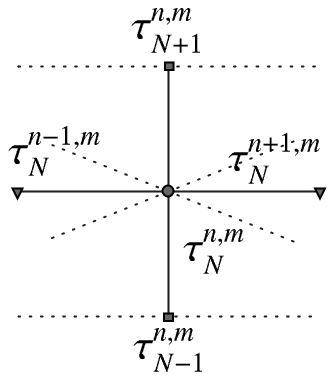}
\caption{Configuration of $\tau$ functions for the bilinear equations of type I. Left:
(\ref{TypeI_1}), center: (\ref{TypeI_2}), right: (\ref{TypeI_3}).}
\label{fig:TypeI}
\end{center}
\end{figure}
\begin{proof}
Application of $T_4=rw_0$ on $\overline{\tau}_0$ yields
\begin{equation}\label{typeI:derivation1}
 T_4(\overline{\tau}_0)
 =c^{-2/3}{a_0}^{-1/3}{a_1}^{-1}{a_2}^{-2/3}
  \frac{\overline{\tau}_0\overline{\tau}_1}{\tau_1}
 +c^{2/3}{a_0}^{1/3}{a_1}^{2/3}a_2
  \frac{\overline{\tau}_0\overline{\tau}_2}{\tau_2}
 +{a_1}^{1/3}{a_2}^{-1/3}
  \frac{\tau_0\overline{\tau}_1\overline{\tau}_2}{\tau_1\tau_2},
\end{equation}
which is rearranged as 
\begin{equation}\label{TypeI_proof_1}
 T_4(\overline{\tau}_0)
 -c^{-2/3}{a_0}^{-1/3}{a_1}^{-1}{a_2}^{-2/3}
  \frac{\overline{\tau}_0\overline{\tau}_1}{\tau_1}
  \left(\frac{q^{1/3}c^{2/3}a_1\tau_0\overline{\tau}_2
  +\overline{\tau}_0\tau_2}{\overline{\tau}_0\tau_2}\right)
  \left(\frac{q^{1/3}c^{2/3}a_2\tau_1\overline{\tau}_0
  +\overline{\tau}_1\tau_0}{\overline{\tau}_1\tau_0}\right)
 +{a_1}^{-2/3}{a_2}^{2/3}\cfrac{{\overline{\tau}_0}^2}{\tau_0}=0.
\end{equation}
Applying $T_2=s_2\pi s_1$ and $T_3=s_2s_1\pi$ on $\overline{\tau}_0$ and $\overline{\tau}_1$,
respectively, we obtain
\begin{align}
 &q^{1/6}c^{1/3}{a_1}^{1/2}\tau_1T_2(\overline{\tau}_0)
 =q^{1/3}c^{2/3}a_1\tau_0\overline{\tau}_2 
 +\overline{\tau}_0\tau_2,
 \label{TypeI_proof_2}\\
 &q^{1/6}c^{1/3}{a_2}^{1/2}\tau_2T_3(\overline{\tau}_1)
 =q^{1/3}c^{2/3}a_2\tau_1\overline{\tau}_0 
 +\overline{\tau}_1\tau_0.
 \label{TypeI_proof_3}
\end{align}
Using (\ref{TypeI_proof_2}) and (\ref{TypeI_proof_3}), we can rewrite (\ref{TypeI_proof_1}) as
\begin{align}
 {T_4}^2(\tau_0)\tau_0
  +{a_1}^{-2/3}{a_2}^{2/3}T_4(\tau_0)^2
  -{a_1}^{-1/6}{a_2}^{1/6}T_2T_4(\tau_0)T_3T_4(\tau_1)=0.
 \label{TypeI:generic}
\end{align}
Then by applying ${T_1}^{l+1}{T_2}^{m}{T_4}^{n-1}$, ${T_1}^{l}{T_2}^{m}{T_4}^{n-1}\pi$, and
${T_1}^{l}{T_2}^{m-1}{T_4}^{n-1}\pi^2$ on (\ref{TypeI:generic}), we obtain
(\ref{TypeI_1}), (\ref{TypeI_2}), and (\ref{TypeI_3}), respectively.
\qed
\end{proof}

Figure \ref{fig:TypeI} shows the configuration of $\tau$ functions in the bilinear equations.
Each bilinear equation takes the form of a linear combination of the three quadratic terms
in $\tau$ functions.
In the left figure, we mark the first, the second, 
and the third multiplication of $\tau$ functions
of (\ref{TypeI_1}) with the square, the circle, and the triangle, respectively.
In the rest of this paper, we use similar representations as above.
\begin{proposition}[Type II: Discrete 2d-Toda type]\label{prop:type2} 
The following bilinear difference equations hold$:$
\begin{align}
 &(1-Q^{-12m}{\alpha_2}^{12})\tau^{n,m}_{N+1}\tau^{n,m}_{N-1}
  +Q^{n-11m}\alpha_0{\alpha_2}^{11}\tau^{n+1,m+1}_{N}\tau^{n-1,m-1}_{N}
  -Q^{n-2m}\alpha_0{\alpha_2}^{2}\tau^{n,m+1}_{N}\tau^{n,m-1}_{N}=0,\label{TypeII_1}\\
 &(1-Q^{12n}{\alpha_0}^{12})\tau^{n,m}_{N+1}\tau^{n,m}_{N-1}
  +Q^{10n+m}{\alpha_0}^{10}{\alpha_2}^{-1}\tau^{n+1,m}_{N}\tau^{n-1,m}_{N}
  -Q^{n+m}\alpha_0{\alpha_2}^{-1}\tau^{n+1,m+1}_{N}\tau^{n-1,m-1}_{N}=0,\label{TypeII_2}\\
 &(1-Q^{12n-12m}{\alpha_0}^{12}{\alpha_2}^{12})\tau^{n,m}_{N+1}\tau^{n,m}_{N-1}
  +Q^{10n-11m}{\alpha_0}^{10}{\alpha_2}^{11}\tau^{n+1,m}_{N}\tau^{n-1,m}_{N}
  -Q^{n-2m}\alpha_0{\alpha_2}^{2}\tau^{n,m+1}_{N}\tau^{n,m-1}_{N}=0.\label{TypeII_3}
\end{align}
\end{proposition}
\begin{figure}[h]
\begin{center}
\includegraphics[width=0.27\textwidth]{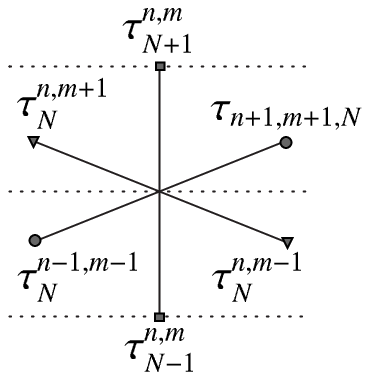}\hspace{1em}
\raise3pt\hbox{\includegraphics[width=0.23\textwidth]{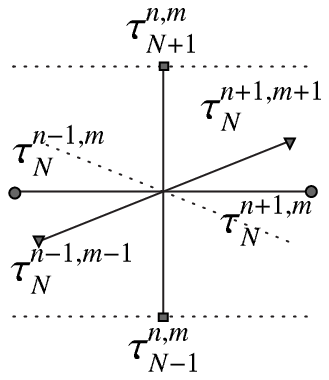}}\hspace{1em}
\raise4pt\hbox{\includegraphics[width=0.23\textwidth]{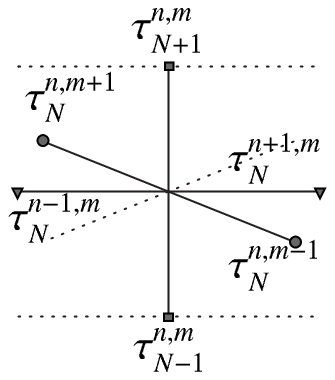}}
\caption{Configuration of $\tau$ functions for the bilinear equations of type II. Left:
(\ref{TypeII_1}), center: (\ref{TypeII_2}), right: (\ref{TypeII_3}).}
\end{center}
\end{figure}
\begin{proof}
Equation (\ref{TypeII_1}) is derived by eliminating $\tau_{l,m,n}$ from (\ref{TypeI_1})
and (\ref{TypeI_2}). We obtain (\ref{TypeII_2}) and (\ref{TypeII_3}) in a similar manner.
\qed
\end{proof}
\begin{proposition}[Type III]\label{prop:type3} The following bilinear equations hold$:$
\begin{align}
 &(Q^{4n-8m+4}{\alpha_1}^{-4}{\alpha_2}^{4}-Q^{4l+4m}{\alpha_0}^{4}{\alpha_2}^{-4})
  \left(\tau^{n,m}_{N}\right)^2
 +Q^{n+m}\alpha_0{\alpha_2}^{-1}\tau^{n+1,m+1}_{N}\tau^{n-1,m-1}_{N}\nonumber\\
 &\hskip120pt -Q^{n-2m+1}{\alpha_1}^{-1}\alpha_2\tau^{n,m+1}_{N}\tau^{n,m-1}_{N}=0,
 \label{TypeIII_1}\\
 &(Q^{4n+4m}{\alpha_0}^{4}{\alpha_2}^{-4}-Q^{-8n+4m-4}{\alpha_0}^{-4}{\alpha_1}^{4})
  \left(\tau^{n,m}_{N}\right)^2
 +Q^{-2n+m-1}{\alpha_0}^{-1}\alpha_1\tau^{n+1,m}_{N}\tau^{n-1,m}_{N}\nonumber\\
 &\hskip120pt -Q^{n+m}\alpha_0{\alpha_2}^{-1}\tau^{n+1,m+1}_{N}\tau^{n-1,m-1}_{N}=0,
 \label{TypeIII_2}\\
 &(Q^{-8n+4m-4}{\alpha_0}^{-4}{\alpha_1}^{4}-Q^{4n-8m+4}{\alpha_1}^{-4}{\alpha_2}^{4})
  \left(\tau^{n,m}_{N}\right)^2
 -Q^{-2n+m-1}{\alpha_0}^{-1}\alpha_1\tau^{n+1,m}_{N}\tau^{n-1,m}_{N}\nonumber\\
 &\hskip120pt +Q^{n-2m+1}{\alpha_1}^{-1}\alpha_2\tau^{n,m+1}_{N}\tau^{n,m-1}_{N}=0.
 \label{TypeIII_3}
\end{align} 
\end{proposition}
\begin{figure}[h]
\begin{center}
\includegraphics[width=0.24\textwidth]{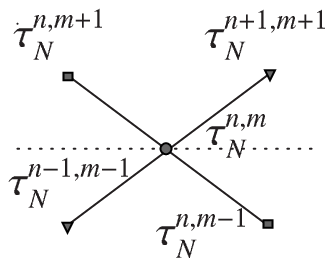}\hspace{1em}
\raise10pt\hbox{\includegraphics[width=0.24\textwidth]{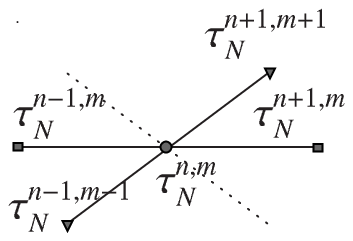}}\hspace{1em}
\includegraphics[width=0.24\textwidth]{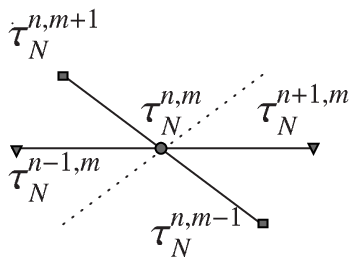}
\caption{Configuration of $\tau$ functions for the bilinear equations of type III. Left:
(\ref{TypeIII_1}), center: (\ref{TypeIII_2}), right: (\ref{TypeIII_3}).}
\end{center}
\end{figure}
\begin{proof}
We obtain (\ref{TypeIII_1}) by eliminating $\tau_{l,m,n+1}\tau_{l,m,n-1}$ from
(\ref{TypeI_1}) and (\ref{TypeI_2}). Other equations can be derived in a similar manner.
\qed
\end{proof}
\begin{proposition}[Type IV]\label{prop:type4}
The following bilinear equation holds$:$
\begin{align}
 Q^{-3n}{\alpha_0}^{-3}(1-Q^{-12m}{\alpha_2}^{-12})\tau^{n+1,m}_{N}\tau^{n-1,m}_{N}
 &-Q^{-3m}{\alpha_2}^{3}(1-Q^{-12l}{\alpha_0}^{-12})\tau^{n,m+1}_{N}\tau^{n,m-1}_{N}\nonumber\\
 &+(Q^{-12m}{\alpha_2}^{-12}-Q^{-12l}{\alpha_0}^{-12})\tau^{n+1,m+1}_{N}\tau^{n-1,m-1}_{N}=0.
 \label{TypeIV}
\end{align} 
\end{proposition}
\begin{figure}[h]
\begin{center}
\includegraphics[width=0.27\textwidth]{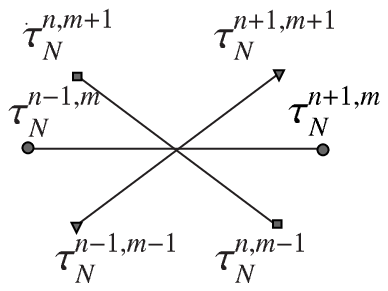}
\caption{Configuration of $\tau$ functions for the bilinear equations of type IV. }
\end{center}
\end{figure}
\begin{proof}
Equation (\ref{TypeIV}) can be derived by eliminating $\tau^{n,m}_{N}$ from 
(\ref{TypeIII_1}) and (\ref{TypeIII_2}).
\qed
\end{proof}
\begin{proposition}[Type V]\label{prop:type5}The following bilinear equations hold$:$
\begin{align}
 &\tau^{n,m}_{N+1}\tau^{n+1,m+1}_{N-1}
 -Q^{n+m-2N}\gamma^{-2}{\alpha_0}^{2}\alpha_1\tau^{n+1,m}_{N}\tau^{n,m+1}_{N}
 -Q^{-2n+2m-4N}\gamma^{4}{\alpha_0}^{-4}{\alpha_1}^{-2}\tau^{n,m}_{N}\tau^{n+1,m+1}_{N}=0,
 \label{TypeV_1}\\
 &\tau^{n+1,m}_{N+1}\tau^{n,m}_{N-1}
 -Q^{-2n+m-2N}\gamma^{-2}{\alpha_0}^{-3}{\alpha_1}^{-1}{\alpha_2}^{-2}
  \tau^{n+1,m+1}_{N}\tau^{n,m-1}_{N}
 -Q^{4n-2m+4N}\gamma^{4}{\alpha_0}^{6}{\alpha_1}^{2}{\alpha_2}^{4}
  \tau^{n+1,m}_{N}\tau^{n,m}_{N}=0,\label{TypeV_2}\\
 &\tau^{n+1,m+1}_{N+1}\tau^{n+1,m}_{N-1}
 -Q^{n-2m-2N+1}\gamma^{-2}{\alpha_1}^{-1}\alpha_2
  \tau^{n,m}_{N}\tau^{n+2,m+1}_{N}
 -Q^{-2n+4m+4N-2}\gamma^{4}{\alpha_1}^{2}{\alpha_2}^{-2}\tau^{n+1,m+1}_{N}
  \tau^{n+1,m}_{N} = 0,\label{TypeV_3}\\
 &\tau^{n+1,m+1}_{N+1}\tau^{n,m}_{N-1}
 -Q^{n+m+2N}\gamma^{2}{\alpha_0}^{2}\alpha_1\tau^{n,m+1}_{N}\tau^{n+1,m}_{N}
 -Q^{-2n-2m-4N}\gamma^{-4}{\alpha_0}^{-4}{\alpha_1}^{-2}\tau^{n+1,m+1}_{N}\tau^{n,m}_{N}=0,
 \label{TypeV_4}\\
 &\tau^{n,m}_{N+1}\tau^{n+1,m}_{N-1}
 -Q^{-2n+m+2N}\gamma^{2}{\alpha_0}^{-3}{\alpha_1}^{-1}{\alpha_2}^{-2}
  \tau^{n,m-1}_{N}\tau^{n+1,m+1}_{N}
 -Q^{4n-2m-4N}\gamma^{-4}{\alpha_0}^{6}{\alpha_1}^{2}{\alpha_2}^{4}
  \tau^{n,m}_{N}\tau^{n+1,m}_{N} = 0,\label{TypeV_5}\\
 &\tau^{n+1,m}_{N+1}\tau^{n+1,m+1}_{N-1}
 -Q^{n-2m+2N}\gamma^{2}\alpha_0{\alpha_2}^{2}\tau^{n+2,m+1}_{N}\tau^{n,m}_{N}
 -Q^{-2n+4m-4N}\gamma^{-4}{\alpha_0}^{-2}{\alpha_2}^{-4}\tau^{n+1,m}_{N}\tau^{n+1,m+1}_{N}=0.
 \label{TypeV_6} 
\end{align} 
\end{proposition}
\begin{figure}[h]
\begin{center}
~\raise20pt\hbox{\includegraphics[width=0.26\textwidth]{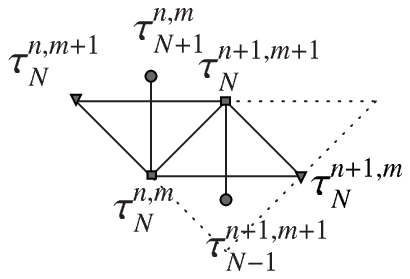}}\hspace{1em}
\includegraphics[width=0.22\textwidth]{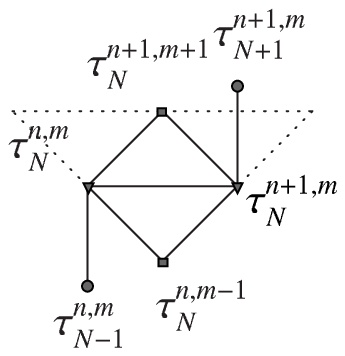}\hspace{1em}
\raise13pt\hbox{\includegraphics[width=0.26\textwidth]{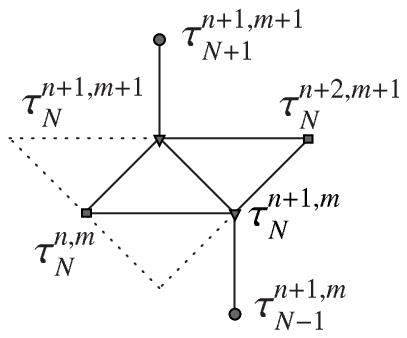}}
\includegraphics[width=0.24\textwidth]{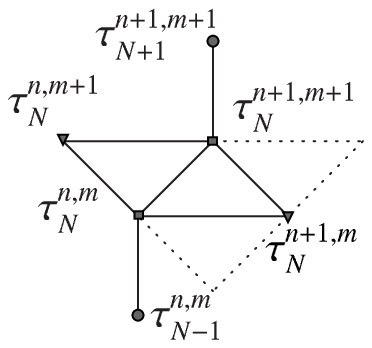}\hspace{1em}
\includegraphics[width=0.23\textwidth]{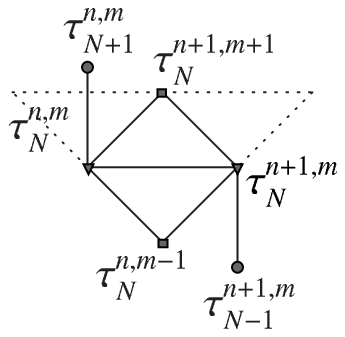}\hspace{1em}
\raise15pt\hbox{\includegraphics[width=0.26\textwidth]{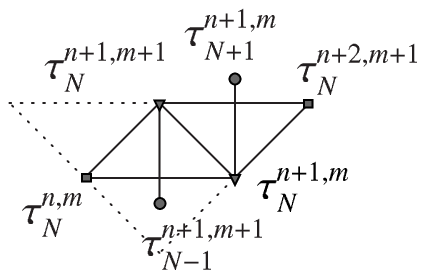}}
\caption{Configuration of $\tau$ functions for the bilinear equations of type V.
Upper left: (\ref{TypeV_1}), upper center: (\ref{TypeV_2}), upper right: (\ref{TypeV_3}),
lower left: (\ref{TypeV_4}), lower center: (\ref{TypeV_5}), lower right: (\ref{TypeV_6}).}
\end{center}
\end{figure}
\begin{proof}
First, we prove (\ref{TypeV_1})--(\ref{TypeV_3}).
We rewrite (\ref{typeI:derivation1}) as
\begin{align}
 T_4(\overline{\tau}_0)
 -c^{-2/3}{a_0}^{-1/3}{a_1}^{-1}{a_2}^{-2/3}
  \frac{\overline{\tau}_1}{\tau_1\tau_2}
  \left(q^{1/3}c^{2/3}a_1\tau_0\overline{\tau}_2
 +\overline{\tau}_0\tau_2\right)
 -c^{2/3}{a_0}^{1/3}{a_1}^{2/3}a_2
  \frac{\overline{\tau}_0\overline{\tau}_2}{\tau_2} = 0.\label{TypeV_proof_1}
\end{align}
By using (\ref{TypeI_proof_2}), we have from (\ref{TypeV_proof_1}) that
\begin{equation}\label{eq:TypeV}
 T_4(\overline{\tau}_0)\tau_2
 -c^{-1/3}{a_0}^{-1/6}{a_1}^{-1/3}{a_2}^{-1/2}\overline{\tau}_1T_2(\overline{\tau}_0)
 -c^{2/3}{a_0}^{1/3}{a_1}^{2/3}a_2\overline{\tau}_0\overline{\tau}_2 = 0,
\end{equation}
which is equivalent to
\begin{equation}\label{TypeV_proof_2}
 {T_1}^{-1}{T_4}^2(\tau_1)T_2(\tau_1) - c^{-1/3}{a_0}^{-1/6}{a_1}^{-1/3}{a_2}^{-1/2}
 T_4(\tau_1){T_1}^{-1}T_2T_4(\tau_1) - c^{2/3}{a_0}^{1/3}{a_1}^{2/3}a_2
 {T_1}^{-1}T_4(\tau_1)T_2T_4(\tau_1) = 0. 
\end{equation}
We obtain (\ref{TypeV_1}), (\ref{TypeV_2}), and (\ref{TypeV_3}) by applying
${T_1}^{l+1}{T_2}^m{T_4}^{n-1}$, ${T_1}^{l+1}{T_2}^m{T_4}^{n-1}\pi$, 
and ${T_1}^{l+1}{T_2}^m{T_4}^{n-1}\pi^2$ on
(\ref{TypeV_proof_2}), respectively.

Next, we prove (\ref{TypeV_4})--(\ref{TypeV_6}).
We rewrite (\ref{typeI:derivation1}) as
\begin{equation}\label{TypeVI_proof_1}
 T_4(\overline{\tau}_0)
 -{a_1}^{1/3}{a_2}^{-1/3}\cfrac{\overline{\tau}_2}{\tau_1\tau_2}
  \left(q^{1/3}c^{2/3}a_2\tau_1\overline{\tau}_0+\overline{\tau}_1\tau_0\right)
 -c^{-2/3}{a_0}^{-1/3}{a_1}^{-1}{a_2}^{-2/3}
  \frac{\overline{\tau}_0\overline{\tau}_1}{\tau_1} = 0.
\end{equation}
By using (\ref{TypeI_proof_3}), we have from (\ref{TypeVI_proof_1}) that
\begin{equation}\label{eq:TypeVI}
 {T_4}^2(\tau_0)\tau_1
  -c^{1/3}{a_0}^{1/6}{a_1}^{1/2}{a_2}^{1/3}T_3T_4(\tau_1)T_4(\tau_2)
  -c^{-2/3}{a_0}^{-1/3}{a_1}^{-1}{a_2}^{-2/3}T_4(\tau_0)T_4(\tau_1) =0 .
\end{equation}
We obtain (\ref{TypeV_4}), (\ref{TypeV_5}), and (\ref{TypeV_6})
by applying ${T_1}^{l+1}{T_2}^m{T_4}^{n-1}\pi^2$, ${T_1}^{l+1}{T_2}^m{T_4}^{n-1}$, and
${T_1}^{l+1}{T_2}^m{T_4}^{n-1}\pi$ on (\ref{eq:TypeVI}), respectively.
\qed
\end{proof}
\begin{proposition}[Type VI]\label{prop:type6} The following bilinear equations hold$:$
\begin{align}
 &\tau^{n,m}_{N+1}\tau^{n+1,m+1}_{N}
 -Q^{-3n+3m+2N-2}\gamma^{2}{\alpha_1}^{3}\tau^{n+1,m}_{N}\tau^{n,m+1}_{N+1}
 +Q^{-6n+6m+4N-4}\gamma^{4}{\alpha_1}^{6}\tau^{n,m}_{N}\tau^{n+1,m+1}_{N+1}=0,
 \label{TypeVI_1}\\
 &\tau^{n+1,m}_{N+1}\tau^{n,m}_{N}
 -Q^{-3m+2N+1}\gamma^{2}{\alpha_2}^{3}\tau^{n+1,m+1}_{N}\tau^{n,m-1}_{N+1}
 +Q^{-6m+4N+2}\gamma^{4}{\alpha_2}^{6}\tau^{n+1,m}_{N}\tau^{n,m}_{N+1}=0,
 \label{TypeVI_2}\\
 &\tau^{n+1,m+1}_{N+1}\tau^{n+1,m}_{N}
 -Q^{3n+2N+4}\gamma^{2}{\alpha_0}^{3}\tau^{n,m}_{N}\tau^{n+2,m+1}_{N+1}
 +Q^{6n+4N+8}\gamma^{4}{\alpha_0}^{6}\tau^{n+1,m+1}_{N}\tau^{n+1,m}_{N+1}=0,
 \label{TypeVI_3}\\
 &\tau^{n+1,m+1}_{N+1}\tau^{n,m}_{N}
 -Q^{-3n+3m-2N-4}\gamma^{-2}{\alpha_1}^{3}\tau^{n+1,m}_{N+1}\tau^{n,m+1}_{N}
 +Q^{-6n+6m-4N-8}\gamma^{-4}{\alpha_1}^{6}\tau^{n+1,m+1}_{N}\tau^{n,m}_{N+1}=0,
 \label{TypeVI_4}\\
 &\tau^{n,m}_{N+1}\tau^{n+1,m}_{N}
 -Q^{-3m-2N-1}\gamma^{-2}{\alpha_2}^{3}\tau^{n+1,m+1}_{N+1}\tau^{n,m-1}_{N}
 +Q^{-6m-4N-2}\gamma^{-4}{\alpha_2}^{6}\tau^{n,m}_{N}\tau^{n+1,m}_{N+1}=0,
 \label{TypeVI_5}\\
 &\tau^{n+1,m}_{N+1}\tau^{n+1,m+1}_{N}
 -Q^{3n-2N+2}\gamma^{-2}{\alpha_0}^{3}\tau^{n,m}_{N+1}\tau^{n+2,m+1}_{N}
 +Q^{6n-4N+4}\gamma^{-4}{\alpha_0}^{6}\tau^{n+1,m}_{N}\tau^{n+1,m+1}_{N+1}=0.
 \label{TypeVI_6}
\end{align} 
\end{proposition}
\begin{figure}[h]
\begin{center}
\includegraphics[width=0.25\textwidth]{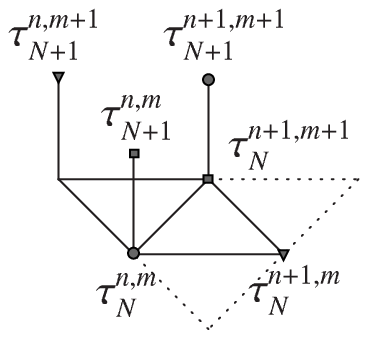}\hspace{1em}
\includegraphics[width=0.25\textwidth]{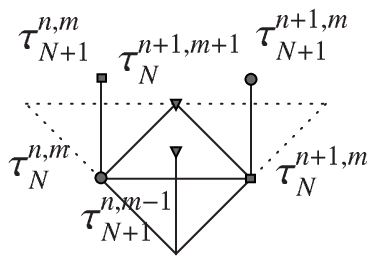}\hspace{1em}
\raise5pt\hbox{\includegraphics[width=0.26\textwidth]{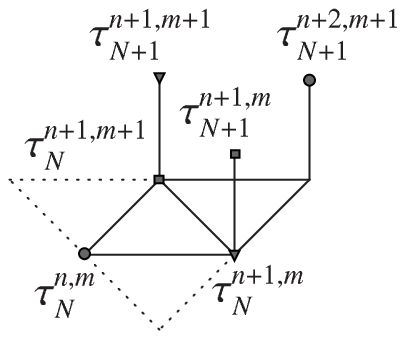}}
\raise5pt\hbox{\includegraphics[width=0.25\textwidth]{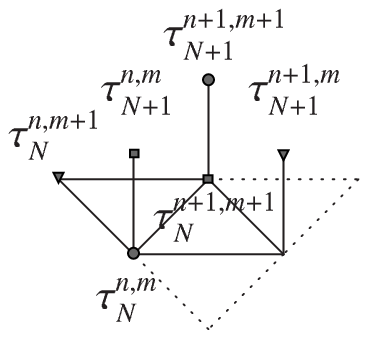}}\hspace{1em}
\includegraphics[width=0.25\textwidth]{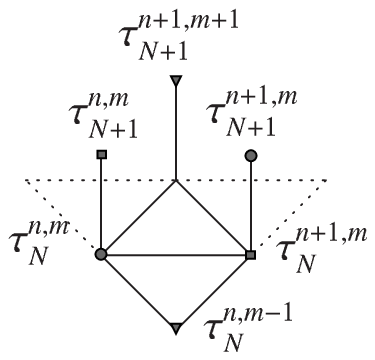}\hspace{1em}
\raise11pt\hbox{\includegraphics[width=0.27\textwidth]{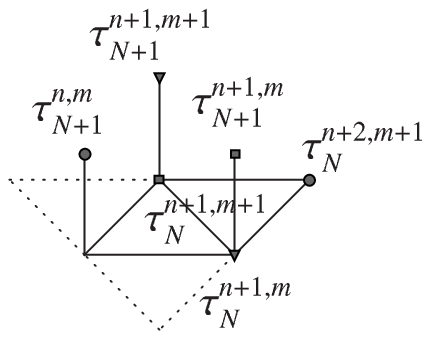}}
\caption{Configuration of $\tau$ functions for the bilinear equations of type VI.
Upper left: (\ref{TypeVI_1}), upper center: (\ref{TypeVI_2}), 
upper right: (\ref{TypeVI_3})
lower left: (\ref{TypeVI_4}), lower center: (\ref{TypeVI_5}),
lower right: (\ref{TypeVI_6}).}
\end{center}
\end{figure}
\begin{proof}
First, we prove (\ref{TypeVI_1})--(\ref{TypeVI_3}).
Equations (\ref{TypeVI_1}), (\ref{TypeVI_2}), and (\ref{TypeVI_3}) can be derived 
by applying ${T_1}^{l+1}{T_2}^m{T_4}^{n}$, 
${T_1}^{l+1}{T_2}^m{T_4}^{n}\pi$, and ${T_1}^{l+1}{T_2}^m{T_4}^{n}\pi^2$
on (\ref{TypeI_proof_3}), respectively.

Next, we prove (\ref{TypeVI_4})--(\ref{TypeVI_6}).
By applying  $T_2$ on $\tau_0$, we obtain
\begin{align}
q^{-1/6}c^{-1/3}{a_1}^{1/2}\overline{\tau}_1T_2(\tau_0)
-q^{-1/3}c^{-2/3}a_1\tau_2\overline{\tau}_0
-\overline{\tau}_2\tau_0=0.\label{eq:TypeVIII}
\end{align}
Equations (\ref{TypeVI_4}), (\ref{TypeVI_5}), and (\ref{TypeVI_6}) can be derived
by applying ${T_1}^{l+1}{T_2}^m{T_4}^{n}$, ${T_1}^{l+1}{T_2}^m{T_4}^{n}\pi$, 
and ${T_1}^{l+1}{T_2}^m{T_4}^{n}\pi^2$
on (\ref{eq:TypeVIII}), respectively.
\qed
\end{proof}
\begin{remark}\rm
The bilinear equations in Proposition \ref{prop:qp3_bl} correspond to
(\ref{TypeVI_1}), (\ref{TypeVI_3}), 
(\ref{TypeVI_4}), (\ref{TypeVI_6}), and (\ref{TypeI_3}). 
\end{remark}
\subsection{Bilinear equations for ${\bm q}$-P$_{\rm\bf II}$}\label{sec:bl_qp2}
The bilinear equations for $q$-P$_{\rm II}$ are derived 
from the equations in Section \ref{sec:bl_qp3}.
Since the parameter space and $\tau$ functions are restricted, we only
have to pick up the bilinear equations that consist 
of the $\tau$ functions on the ``unit-strip,''
and to rewrite them in terms of $R_1$ instead of $T_1$ (see Figure
\ref{fig:R1_strip}). Therefore, only the bilinear equations of type V and VI are
relevant. We use the notation in (\ref{notation:tau_qp2}).
\begin{proposition} \label{prop:tau_bl_qp2}
The following bilinear equations hold$:$
\begin{align}
 &\tau^{k+1}_{N+1}\tau^{k+2}_{N-1}
 -Q^{(k-4N+2)/2}\gamma^{-2}\alpha_0\tau^{k+3}_{N}\tau^{k}_{N}
 -Q^{-k+4N-2}\gamma^{4}{\alpha_0}^{-2}\tau^{k+1}_{N}\tau^{k+2}_{N}=0,
 \label{Prop_B2_1}\\
 &\tau^{k+2}_{N+1}\tau^{k+1}_{N-1}
 -Q^{(k+4N+2)/2}\gamma^{2}\alpha_0\tau^{k+3}_{N}\tau^{k}_{N}
 -Q^{-k-4N-2}\gamma^{-4}{\alpha_0}^{-2}\tau^{k+2}_{N}\tau^{k+1}_{N}=0,
 \label{Prop_B2_2}\\
 &Q^{-(3k-4N+4)/2}\gamma^{2}{\alpha_0}^{-3}\tau^{k+3}_{N}\tau^{k}_{N+1}
 -Q^{-3k+4N-4}\gamma^{4}{\alpha_0}^{-6}\tau^{k+1}_{N}\tau^{k+2}_{N+1}
 -\tau^{k+1}_{N+1}\tau^{k+2}_{N}=0,
 \label{Prop_B2_3}\\
 &Q^{-(3k+4N+8)/2}\gamma^{-2}{\alpha_0}^{-3}\tau^{k+3}_{N+1}\tau^{k}_{N}
 -Q^{-3k-4N-8}\gamma^{-4}{\alpha_0}^{-6}\tau^{k+2}_{N}\tau^{k+1}_{N+1}
 -\tau^{k+2}_{N+1}\tau^{k+1}_{N}=0.
 \label{Prop_B2_4}
\end{align} 
\end{proposition}
\begin{proof}
Noticing (\ref{notation:tau_qp2_2}), we obtain from (\ref{eq:TypeV})
\begin{equation}
 {R_1}^{-2}{T_4}^2(\tau_1){R_1}^{-1}(\tau_1)
 -q^{-5/12}c^{-1/3}{a_0}^{1/6}
  T_4(\tau_1){R_1}^{-3}T_4(\tau_1)
 -q^{5/6}c^{2/3}{a_0}^{-1/3}
  {R_1}^{-2}T_4(\tau_1){R_1}^{-1}T_4(\tau_1)=0,
\end{equation}
from which (\ref{Prop_B2_1}) is derived by applying ${R_1}^{m+3}{T_4}^{n-1}$. Similarly, we have
\begin{equation}\label{Prop_B2_proof}
 {T_4}^2(\tau_1){R_1}^{-1}(\tau_1)
 -q^{1/3}c^{1/3}{a_0}^{1/6}
  R_1T_4(\tau_1){R_1}^{-2}T_4(\tau_1)
 -q^{-2/3}c^{-2/3}{a_0}^{-1/3}
  T_4(\tau_1){R_1}^{-1}T_4(\tau_1)=0.
\end{equation}
by applying $\pi$ on (\ref{eq:TypeVI}). Then we obtain (\ref{Prop_B2_2})
by applying ${R_1}^{m+2}{T_4}^{n-1}$ on (\ref{Prop_B2_proof}).
Equation (\ref{Prop_B2_3}) is derived by applying ${R_1}^{m+3}{T_4}^{n}$ on 
\begin{equation}
 q^{1/6}c^{1/3}{a_0}^{-1/2}\tau_1{R_1}^{-3}T_4(\tau_1)
 -q^{1/3}c^{2/3}{a_0}^{-1}{R_1}^{-2}(\tau_1){R_1}^{-1}T_4(\tau_1)
 -{R_1}^{-2}T_4(\tau_1){R_1}^{-1}(\tau_1)=0, 
\end{equation}
which follows from (\ref{TypeI_proof_2}). Finally, we obtain (\ref{Prop_B2_4}) by applying
 ${R_1}^{m+3}{T_4}^{n}$ on
\begin{equation}
 q^{-1/6}c^{-1/3}{a_0}^{-1/2}T_4(\tau_1){R_1}^{-3}(\tau_1)
 -q^{-1/3}c^{-2/3}{a_0}^{-1}{R_1}^{-1}(\tau_1){R_1}^{-2}T_4(\tau_1)
 -{R_1}^{-1}T_4(\tau_1){R_1}^{-2}(\tau_1)=0,
\end{equation}
which is follows from (\ref{eq:TypeVIII}).
\qed
\end{proof}
\begin{figure}[h]
\begin{center}
\includegraphics[width=0.2\textwidth]{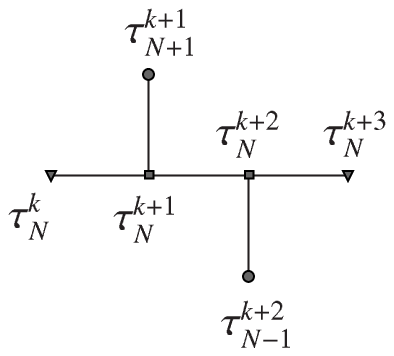}\hspace{1em}
\includegraphics[width=0.2\textwidth]{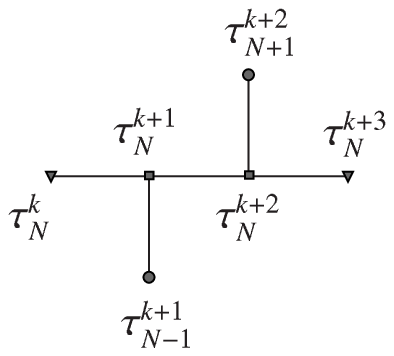}\hspace{1em}
\includegraphics[width=0.2\textwidth]{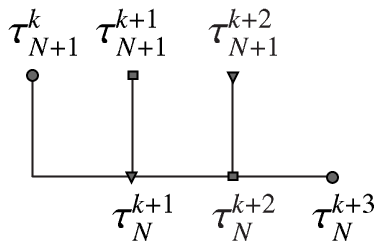}\hspace{1em}
\includegraphics[width=0.2\textwidth]{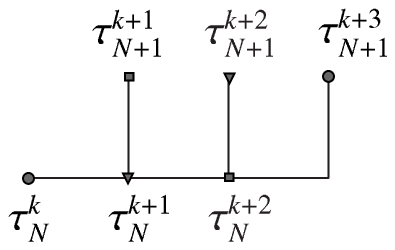}
\caption{Configuration of $\tau$ functions for the bilinear equations in Proposition
\ref{prop:tau_bl_qp2}. The figures correspond to (\ref{Prop_B2_1}), (\ref{Prop_B2_2}),
(\ref{Prop_B2_3}), and (\ref{Prop_B2_4}), respectively, from the left to the right.}
\end{center}
\end{figure}
\begin{remark}\rm
The bilinear equations in Proposition \ref{prop:qp2_bl} correspond to
(\ref{Prop_B2_3}), (\ref{Prop_B2_4}), and (\ref{Prop_B2_1}).
\end{remark}

K. Kajiwara: Faculty of Mathematics, Kyushu University, 
744 Motooka, Fukuoka 819-0395, Japan\\
E-mail address: kaji@math.kyushu-u.ac.jp\\[1em]
N. Nakazono: Graduate School of Mathematics, Kyushu University, 
744 Motooka, Fukuoka 819-0395, Japan\\
E-mail address: n-nakazono@math.kyushu-u.ac.jp\\[1em]
T. Tsuda: Faculty of Mathematics, Kyushu University, 
744 Motooka, Fukuoka 819-0395, Japan\\
E-mail address: tudateru@math.kyushu-u.ac.jp
\end{document}